\documentclass[aos]{imsart}

\RequirePackage{amsthm,amsmath,amsfonts,amssymb}
\RequirePackage[numbers]{natbib}
\RequirePackage[colorlinks,citecolor=blue,urlcolor=blue]{hyperref}
\RequirePackage{graphicx}
\RequirePackage{booktabs}
\RequirePackage[flushleft]{threeparttable}
\usepackage{float}
\usepackage{tikz}
\newcommand*\circled[1]{\tikz[baseline=(char.base)]{\node[shape=circle,draw,inner sep=2pt] (char) {#1};}}

\usepackage{multirow}
\usepackage{xcolor}
\usepackage[linesnumbered,ruled,vlined]{algorithm2e}
\usepackage{multicol,caption}
\usepackage{subcaption}

\usepackage{array, ltablex, multirow}
\usepackage{booktabs}
\newcolumntype{Y}{>{\centering\arraybackslash}X}

\usepackage{mathtools}
\usepackage[utf8]{inputenc}

\newcommand{\E}{\mathbb{E}}

\newcommand{\boldB}{\boldsymbol{B}}

\newcommand{\boldX}{\boldsymbol{X}}

\newcommand{\boldf}{\boldsymbol{f}}

\newcommand{\boldu}{\boldsymbol{u}}
\newcommand{\boldv}{\boldsymbol{v}}
\newcommand{\boldw}{\boldsymbol{w}}
\newcommand{\boldx}{\boldsymbol{x}}

\newcommand{\boldc}{\boldsymbol{c}}
\newcommand{\boldp}{\boldsymbol{p}}

\newcommand{\bmu}{\boldsymbol{\mu}}

\newcommand{\bSigma}{\boldsymbol{\Sigma}}

\startlocaldefs
\newtheorem{assumption}{Assumption}

\newtheorem{theorem}{Theorem}[section]
\newtheorem{lemma}[theorem]{Lemma}
\newtheorem{corollary}[theorem]{Corollary}
\theoremstyle{remark}
\newtheorem{definition}[theorem]{Definition}

\endlocaldefs

\begin{document}
	
	\begin{frontmatter}
		\title{Variational Bayes Algorithm and Posterior Consistency of Ising Model Parameter Estimation}
		\runtitle{Variational Bayes for Ising Models}
		
		\begin{aug}
			\author[]{\fnms{Minwoo} \snm{Kim}\ead[label=e1,mark]{kimminw3@msu.edu}},
			\author[]{\fnms{Shrijita} \snm{Bhattacharya}\ead[label=e2,mark]{bhatta61@msu.edu}}
			\and
			\author[]{\fnms{Tapabrata} \snm{Maiti}\ead[label=e3,mark]{maiti@msu.edu}}
			\address[]{Department of Statistics and Probability,
				Michigan State University,
				\printead{e1,e2,e3}}
			
		\end{aug}
		
		\begin{abstract}
			Ising models originated in statistical physics and are widely used in modeling spatial data and computer vision problems. However, statistical inference of this model remains challenging due to intractable nature of the normalizing constant in the likelihood. Here, we use a pseudo-likelihood instead to study the Bayesian estimation of two-parameter, inverse temperature, and magnetization, Ising model with a fully specified coupling matrix.  We develop a computationally efficient variational Bayes procedure for model estimation.  Under the Gaussian mean-field variational family, we derive posterior contraction rates of the variational posterior obtained under the pseudo-likelihood. We also discuss the loss incurred due to variational posterior over true posterior for the pseudo-likelihood approach. Extensive simulation studies validate the efficacy of mean-field Gaussian and bivariate Gaussian families as the possible choices of the variational family for inference of Ising model parameters.
		\end{abstract}
		
		
		
	\end{frontmatter}
	
	\section{Introduction}
	A popular way of modeling a spatially dependent binary vector $\boldx = (x_1, \dots, x_n)^\top $ is to take advantage of Ising model named after the physicist Ernst Ising \cite{ising1925beitrag} which has been used in a wide range of applications including spatial data analysis and computer vision. Many different versions of Ising model have emerged in the literature. In this paper, we focus on two-parameter Ising model, which has an inverse temperature parameter $\beta > 0$ and a magnetization parameter $B \neq 0$, with a symmetric coupling matrix $A_n \in \mathbb{R}^{n \times n}$. An Ising model is often represented by an undirected graph in which each vertex represents $x_i \in \{-1, 1\}$ and the connections between $x_i$'s are determined by $A_n$. Here, $\beta$ characterizes the strength of interactions among $x_i$'s and $B$ represents external influence on $\boldx$. In the first place, Ising model has been introduced for the relations between atom spins \cite{brush1967history} with the domain $\{-1,1\}^n$. While we work with the domain $\{-1,1\}^n$, in many current applications, Ising model has been defined with different domain $\{0,1\}^n$. One can read \cite{haslbeck2020interpreting} for more details on two different domains.  
	
	Estimation of Ising model parameters has received considerable attention in statistics and computer science literature. The existing literature can be broadly divided into two groups. Some  literature assume that i.i.d. (independently and identically distributed) copies of data ($\boldx$ vector) are available for inference, \cite{anandkumar2012high}, \cite{bresler2015efficiently}, \cite{lokhov2018optimal}, \cite{ravikumar2010high}, and \cite{xue2012nonconcave}. Another category of literature assumes that only one sample is observable,  \cite{bhattacharya2018inference}, \cite{chatterjee2007estimation}, \cite{comets1992consistency}, \cite{comets1991asymptotics},
	\cite{Ghosal:2020},
	\cite{gidas1988consistency}, and \cite{guyon1992asymptotic}.  Under the assumption of only one observation, \cite{comets1991asymptotics} showed that the MLE of $\beta > 0$ for Curie-Weiss model is consistent if $B \neq 0$ is known, and vice versa. They also proved that the joint MLE does not exist when neither $\beta$ nor $B$ is given. In this regard, \cite{Ghosal:2020} addressed joint estimation of $(\beta, B)$ using pseudo-likelihood and showed that the pseudo-likelihood estimator is consistent under some conditions on coupling matrix $A_n$. 
	
	Some previous studies which utilized Bayesian methodology along with Ising model include those of \cite{li2010bayesian} which explored Bayesian variable selection with Ising prior to capture structural information. \cite{li2015spatial} proposed a joint Ising and Dirichlet Process (Ising-DP) prior for simultaneous selection and clustering in Bayesian framework. In their PhD thesis, \cite{zhou2009bayesian} dealt with estimation of parameters of Ising and Potts models using Bayesian approach. All these  works assume replicated data and do not provide theoretical validation. In this article, we provide a Bayesian estimation methodology for  model parameters in an Ising model when one observes the data only once.
	
	{\em Methodological Contribution:}
	One of the main challenges in the Bayesian estimation of Ising models lies in the intractable nature of the normalizing constant in the likelihood. Following the works of \cite{Ghosal:2020}, \cite{bhattacharya2018inference} and \cite{okabayashi2011extending}, we replace the true likelihood of the Ising model by a pseudo-likelihood. As a first contribution, we establish that the posterior based on the pseudo-likelihood is consistent for a suitable choice of the prior distribution. Further, we use variational Bayes (VB) approach which has recently become a popular and computationally powerful alternative to MCMC. In order to approximate the unknown posterior distribution using VB, we need to choose an appropriate variational family. We propose a Gaussian mean field family and general bivariate normal family with transformation of the parameters to $(\log\beta, B)$. For implementation of VB, we employ a black box variational inference (BBVI), \cite{Ranganath:2014}. In BBVI, we need  to evaluate the likelihood to compute the gradient estimate. But the existence of an unknown normalizing constant in likelihood of Ising model prevents us using BBVI directly. So, we use pseudo-likelihood as in \cite{Ghosal:2020}. Replacing the true likelihood of Ising model with pseudo-likelihood, we are able to compute all the quantities needed for implementing BBVI. 
	
	{\em Theoretical Contribution:} 
	The main theoretical contribution of this work lies in establishing the  consistency of the variational posterior for the Ising model with the true likelihood replaced by the pseudo-likelihood. In this direction, we first establish the rates at which the true posterior based on the pseudo-likelihood concentrates around the $\varepsilon_n$- shrinking neighborhoods of the true parameter. With a suitable bound on the Kulback-Leibler distance between the true and the variational posterior, we next establish the rate of contraction for the variational posterior and demonstrate that the variational posterior also concentrates around $\varepsilon_n$-shrinking neighborhoods of the true parameter. These results have been derived under three set of assumptions on the coupling matrix $A_n$ (see section \ref{section3} for more details). Indeed, we demonstrate that the variational posterior consistency holds for the same set of assumptions on $A_n$ as those needed for the convergence of the maximum likelihood estimates based on the pseudo-likelihood. One of the main caveats in establishing the posterior contraction rates under the pseudo-likelihood structure is in ensuring that the concentration of the variational posterior occurs in $\mathbb{P}_0^{(n)}$ probability where $\mathbb{P}_0^{(n)}$ is the distribution induced by the true likelihood and not the pseudo-likelihood. Indeed, we could show that in $\mathbb{P}_0^{(n)}$ probability, the contraction of variational posterior happens at the rate $1-1/M_n$ in contrast to the faster  rate $1-\exp(-Cn\varepsilon_n^2), C>0$ for the true posterior. As a final theoretical contribution, we establish that the variational Bayes estimator convergences to the true parameters at the rate $1/\varepsilon_n$ where $\varepsilon_n$ can be chosen $n^{-\delta}$, $0<\delta<1/2$ provided the $A_n$ matrix satisfies certain regularity assumptions.
	
	The rest of the paper is organized as follows: Section \ref{section2} defines the likelihood and pseudo-likelihood of Ising model with two parameters $(\beta, B)$ and provides the details of our Bayesian estimation using variational inference approach. In Section \ref{section3} we discuss our main theoretical developments and sketch of the proof. The details of the proof are deferred to the Appendix. The numerical studies are provided in Section \ref{section4}. We give a comparison of our variational Bayes estimates to existing maximum likelihood estimators based on pseudo-likelihood.

	\section{Model and methods}
	\label{section2}
	\subsection{Ising model}
	For a representation of an Ising model with two parameters $\beta > 0$ and  $B \neq 0$, consider an undirected graph which has $n$ vertices $x_i$, $i=1,\dots,n$. Each vertex of the graph takes a value either -1 or 1, i.e., $x_i \in \{-1,1\}$. Then, we define a likelihood of Ising model as the probability of the vector $\boldx = (x_1, \dots, x_n)^\top \in \{-1,1\}^n$:
	\begin{align}
		\label{likel}
		\mathbb{P}_{\beta,B}^{(n)}(\boldX = \boldx) = \frac{1}{Z_n(\beta, B)} \exp \left(\frac{\beta}{2} \boldx^\top A_n \boldx + B\sum_{i=1}^n x_i \right),
	\end{align}
	where $Z_n(\beta,B)$ is a normalizing constant which makes sum of \eqref{likel} over the support $\{-1,1\}^n$ equal to 1. $A_n$ is a coupling matrix of size $n \times n$ which determines the connections between the coordinates of $\boldx$.  More precisely, let $\mathcal{E} = \{(i, j) \mid i \sim j,~ 1 \leq i,j \leq n \}$ be the set of edges in the graph where $i \sim j$ denote that the vertices $i$ and $j$ are connected. Then $A_n$ is a symmetric matrix with $A_n(i,j) = 0$ for all $(i,j) \notin \mathcal{E}$ and $A_n(i,j) > 0$ for all $(i,j) \in \mathcal{E}$.
	
	\noindent {\it For the purpose of this paper, we work with a scaled adjacency matrix $A_n$ whose all diagonal elements are  zeros and the other elements are non-negative. }
	\begin{definition}[Scaled  adjacency  matrix] A scaled adjacency matrix for a graph $G_n$ with $n$ vertices is defined as:
		\begin{align}
			A_n(i,j) := \begin{cases} \frac{n}{2|G_n|} & \mbox{if } (i,j) \in \mathcal{E}  \\ 0 & \mbox{otherwise.} \end{cases}, \nonumber
		\end{align}
		where $|G_n|$ denotes the number of edges.
	\end{definition}
	In our study, we assume that only one $\boldx \in \{-1, 1\}^n$ is observed. In this regard, estimating all the elements of $A_n$ is impossible because $A_n$ has $n(n-1)/2$ distinct values. In this work, we primarily focus on the problem of estimation of the  parameters $(\beta, B)$ under the assumption of a fully known coupling matrix $A_n$. The same set up considered by \cite{bhattacharya2018inference}, \cite{Ghosal:2020}, and \cite{okabayashi2011extending}.
	
	\subsection{Pseudo-likelihood}
	It is challenging to use the likelihood \eqref{likel} directly because of the unknown normalizing constant $Z_n(\beta, B)$. Due to the intractable nature of the likelihood, the standard Bayesian implementation is computationally intractable. We thereby propose the use of the conditional probability of $x_i$ given others. It is easily calculated because $x_i$ is binary:
	\begin{align*}
		\mathbb{P}_{\beta,B}^{(n)}\big(X_i = 1 | X_j, j\neq i \big) = \frac{e^{\beta m_i(\boldx) + B}}{e^{\beta m_i(\boldx) + B} + e^{-\beta m_i(\boldx) - B}},
	\end{align*}
	where $m_i(\boldx) = \sum_{j=1}^n A_n(i,j)x_j$. 
	
	The pseudo-likelihood of Ising model corresponding to the likelihood in \eqref{likel}, is defined as the product of  one dimensional conditional distributions (see \cite{Ghosal:2020}, for further details):
	\begin{align}
		\label{pseudo}
		\nonumber   & \prod_{i=1}^n \mathbb{P}_{\beta,B}^{(n)}\left(X_i = x_i \mid X_j, j\neq i \right)
		\\
		&= 2^{-n}\exp  \left(\sum_{i=1}^n \left(\beta x_i m_i(\boldx) + Bx_i 
		- \log\cosh(\beta m_i(\boldx)+B) \right)\right).
	\end{align}
	Our subsequent Bayesian development will make use of the pseudo-likelihood  \eqref{pseudo} instead of the true likelihood  \eqref{likel}. We shall establish that the posterior obtained by the use of the pseudo-likelihood allows for  consistent estimation of the model parameters $\beta$ and $B$.
	
	\subsection{Bayesian formulation}
	Let $\theta = (\beta,B)$ be the parameter set of interest. We consider the following  independent prior distribution $p(\theta) = p_\beta(\beta)p_B(B)$, with $p_\beta(\beta)$ as a log-normal prior for $\beta$ and $p_B(B)$ as a normal prior for $B$ as follows:
	\begin{align}
		\label{e:prior}
		p_\beta(\beta) &= \frac{1}{\beta\sqrt{2\pi}} e^{-\frac{(\log\beta)^2}{2}}, \hspace{5mm}
		p_B(B) = \frac{1}{\sqrt{2\pi}}e^{-\frac{B^2}{2}}.
	\end{align}
	The assumption of log-normal prior on $\beta$ is to ensure the positivity of $\beta$. Let  $L(\theta)$ be the pseudo-likelihood function given by \eqref{pseudo}, then the above prior structure  leads to the following posterior distribution
	\begin{align}
		\label{true_po}
		\Pi(\mathcal{A} \mid X^{(n)}) =   \frac{\int_{\mathcal{A}}\pi(\theta,X^{(n)})d\theta}{m(X^{(n)})}=\frac{\int_{\mathcal{A}}L(\theta)p(\theta)d\theta}{\int L(\theta)p(\theta)d\theta} ,
	\end{align}
	for any set $\mathcal{A} \subseteq \Theta$ where $\Theta$ denotes the parameter space of $\theta$. Note,  $\pi(\theta,X^{(n)})$ is the  joint density of $\theta$ and the data $X^{(n)}$ and  $m(X^{(n)})$ is the marginal density of $X^{(n)}$ which is free from the parameter set $\theta$.
	
	\subsection{Variational inference}
	\label{sec:var-inf}
	
	Next, we need to provide a variational approximation to the posterior distribution  \eqref{true_po}. In this direction, we consider two choices of the variational family to obtain approximated posterior distribution. One candidate of our variational family, for the virtue of simplicity, is a mean-field (MF) Gaussian family:
	\begin{align}
		\label{mf:family}
		\mathcal{Q}^{\bf MF} = \bigg\{q(\theta)\mid~ q(\theta) = q_\beta(\beta)q_B(B), \log \beta \sim N(\mu_1, \sigma_1^2), B \sim N(\mu_2, \sigma_2^2)\bigg\}, 
	\end{align}
	The above variational family is the same as a lognormal distribution on $\beta$ and normal distribution on $B$. Also, $\beta$ and $B$ are independent in $\mathcal{Q}^{\bf MF}$ and each $q(\theta) \in \mathcal{Q}^{\bf MF}$ is governed by its own parameter set $\nu^{\bf MF} = (\mu_1, \mu_2, \sigma_1^2, \sigma_2^2)^\top$. $\nu^{\bf MF}$ denote the set of variational parameters which will be updated to find the optimal variational distribution closest to the true posterior. 
	
	Beyond the mean field family, we suggest a  bivariate normal (BN) family to exploit the interdependence among the parameters $(\beta, B)$:
	\begin{align}
		\label{bn:family}
		\mathcal{Q}^{\bf BN} = \Bigg\{ q(\theta)\mid~ q(\theta)=q(\beta,B), (\log \beta,B) \sim MVN(\bmu,\bSigma),   \boldsymbol{\mu} = \begin{pmatrix}
			\mu_1  \\ \mu_2 
		\end{pmatrix}, \bSigma =\begin{pmatrix}
			\sigma_{11}&\sigma_{12}  \\ \sigma_{12}&\sigma_{22} 
		\end{pmatrix}\Bigg\}
	\end{align}
	$\mathcal{Q}^{\bf MF}$ can also be represented as (independent) bivariate normal family. The variational parameters of BN family are $\nu^{\bf BN} = (\mu_1, \mu_2, \sigma_{11}, \sigma_{22}, \sigma_{12})^\top$.
	
	Once a variational family is chosen, one can  find the variational posterior by minimizing the Kullback-Leibler (KL) divergence between a variational distribution $q \in \mathcal{Q}$
	and the true posterior  \eqref{true_po}. The variational posterior is thus given by
	\begin{align}
		\label{e:var-posterior}
		Q^* = \underset{Q \in \mathcal{Q}}{\arg\min}~ {\rm KL}(Q, \Pi( \mid X^{(n)})),
	\end{align}
	where ${\rm KL}(Q, \Pi( \mid X^{(n)}))$ is the KL divergence given by
	\begin{align*}
		{\rm KL}(Q, \Pi( \mid X^{(n)})) = \int \log ( q(\theta) / \pi(\theta \mid X^{(n)})) q(\theta) d\theta,
	\end{align*}
	where $q$ and $\pi(\mid X^{(n)})$ are the densities corresponding to $Q$ and $\Pi(\mid X^{(n)})$ respectively. Based on $\eqref{true_po}$, we re-write the KL divergence as:
	\begin{align*}
		{\rm KL}(Q, \Pi(\mid X^{(n)})) 
		&= \int (\log q(\theta) - \log \pi(\theta, X^{(n)}) ) q(\theta) d\theta + \log m(X^{(n)}) \\
		&= -{\rm ELBO}(Q, \Pi(, X^{(n)})) + \log m(X^{(n)}).
	\end{align*}
	The first term is the negative Evidence Lower Bound (ELBO) and observe that the second term does not depend on $q$. Therefore, minimizing KL divergence is equivalent to maximizing the ELBO. So, we search  for an optimal $q$ by maximizing the ELBO:
	\begin{align*}
		Q^* = \underset{Q \in \mathcal{Q}}{\arg\max}~ {\rm ELBO}(Q, \Pi( ,X^{(n)})).
	\end{align*}
	To optimize the ELBO, we consider the ELBO as a function of variational parameters $\nu$:
	\begin{align*}
		\mathcal{L}(\nu) := \E_Q( \log \pi(\theta, X^{(n)}) - \log q(\theta; \nu) ).
	\end{align*}
	\cite{Ranganath:2014} suggested black box variational inference (BBVI) to optimize $\mathcal{L}(\nu)$ using gradient descent method. The gradient of $\mathcal{L}(\nu)$ is:
	\begin{align}
		\label{grad_elbo}
		\nabla_{\nu}\mathcal{L} &= \nabla_{\nu} \E_Q ( \log \pi(\theta, X^{(n)})  - \log q(\theta; \nu) )  \nonumber\\
		&= \int q(\theta;\nu) \nabla_{\nu} \log q(\theta;\nu) ( \log \pi(\theta, X^{(n)}) - \log q(\theta;\nu) ) d\theta \nonumber\\
		&\enskip + \int q(\theta;\nu) \nabla_{\nu} ( \log \pi(\theta, X^{(n)}) - \log q(\theta;\nu) ) d\theta  \nonumber\\
		&= \E_Q (\nabla_{\nu} \log q(\theta;\nu) (\log \pi(\theta, X^{(n)}) - \log q(\theta;\nu))).
	\end{align}
	The last equality holds because $\E_Q(\nabla_{\nu} \log q(\theta;\nu)) = 0$ and $ \nabla_{\nu}  \log \pi(\theta, X^{(n)}) =0$. We cannot exactly compute the expectation form \eqref{grad_elbo} of the gradient $\nabla_{\nu}\mathcal{L}$, which lead us to employ Monte Carlo estimates:
	\begin{align}
		\label{mc_grad}
		\widehat{\nabla}_{\nu}\mathcal{L} = \frac{1}{S}\sum_{s=1}^S \nabla_{\nu} \log q(\theta^{(s)}; \nu) (\log \pi(\theta^{(s)}, X^{(n)}) - \log q(\theta^{(s)}; \nu)),
	\end{align}
	where $\theta^{(1)}, \dots, \theta^{(S)}$ are samples generated from $q(\theta;\nu)$. Using the estimate $\eqref{mc_grad}$, we iteratively update $\nu$ in the direction of increasing the objective function $\mathcal{L}(\nu)$. The summary of BBVI algorithm is shown in Algorithm \ref{algo1}:
	
	\begin{algorithm}[H]
		\label{algo1}
		\SetAlgoLined
		\KwResult{Optimal variational parameters $\nu^*$ }
		Initialize $p(\theta)$, $q(\theta; \nu^1)$ and learning rate sequence $\rho_t$. \\
		\While{{\rm ELBO} increases}{
			Draw $\theta^{(s)} \sim q(\theta; \nu^t)$, $s=1,\dots, S$\;
			Get $\widehat{\nabla}_{\nu} \mathcal{L}$ based on the $S$ sample points \;
			Update $\nu^{t+1} \leftarrow \nu^t + \rho_t \widehat{\nabla}_{\nu} \mathcal{L}$ \; }
		\caption{BBVI}
	\end{algorithm}
	In Algorithm \ref{algo1},  $\rho_t$, $t=1, 2, \cdots$ denotes a sequence of learning rates which satisfy the Robbin-Monro conditions \cite{Robbins:1951}, i.e. 
	$    \sum_{t=1}^{\infty}\rho_t = \infty$ and $\quad \sum_{t=1}^{\infty}\rho_t^2 < \infty$.
	Also, let $\sigma \in \nu$ be a variational parameter which must be positive. During the updating procedure, it may occur that $\sigma$ takes a negative value. In order to preclude this issue, we consider a reparametrization $\sigma = \log(1+e^\eta)$ and update the quantity $\eta$, as a free parameter,  instead of updating $\sigma$. We address more details of BBVI algorithm implementation in Supplement  Section A.1.
	
	\section{Main Theoretical Results}
	\label{section3}
	
	In this section, we establish the posterior consistency of the variational posterior  \eqref{e:var-posterior}. In this direction, we establish the variational posterior contraction rates to evaluate how well the posterior distribution of $\beta$ and $B$ under the variational approximation concentrates around the true values $\beta_0$ and $B_0$. Towards the proof, we make the following assumptions
	\begin{assumption}[Bounded row sums of $A_n$]
		\label{assumption_rs}
		The row sums of $A_n$ are bounded above
		$$   \max_{i \in [n]} \sum_{j=1}^n A_n(i,j) \leq \gamma, $$
	\end{assumption}
	for a constant $\gamma$ independent of $n$. Assumption \ref{assumption_rs} is the same as (1.2) in \cite{Ghosal:2020}. As a consequence of assumption 1, it can be shown $m_i(\boldx) \leq \gamma$, $i = 1 \dots, n$.
	\begin{assumption}[Mean field assumption on $A_n$]
		\label{assumption_mf}
		Let $\epsilon_n \rightarrow 0$ and $n\epsilon_n^2 \rightarrow \infty$ such that
		$$(i) \:\:\sum_{i=1}^n \sum_{j=1}^n A_n(i,j)=O(n\epsilon_n^2), \hspace{10mm} (ii)\:\: \sum_{i=1}^n \sum_{j=1}^n A_n(i,j)^2 =o(n\epsilon_n^2). $$
	\end{assumption}
	Assumption \ref{assumption_mf}-$(i)$ is the same as  condition (1.4) in \cite{Ghosal:2020} on $A_n$ for $\epsilon_n=1$.  Assumption \ref{assumption_mf}-$(ii)$ is the same as (1.6) in \cite{Ghosal:2020} with $\epsilon_n=1$. For more details on the mean field assumption, we refer to Definition 1.3 in  \cite{Basak:2017}.
	\begin{assumption}[Bounded variance of $A_n$]
		\label{assumption_vb} Let $\bar{A}_n=(1/n)\sum _{i=1}^n \sum_{j=1}^n A_n(i,j)$,
		$$\liminf_{n \to \infty}    \frac{1}{n} \sum_{i=1}^n (\sum_{j=1}^n A_n(i,j)-\bar{A}_n)^2>0.$$
	\end{assumption}
	Finally, the Assumption \ref{assumption_vb} corresponds to (1.7) in \cite{Ghosal:2020}. The validity of Assumption \ref{assumption_vb} ensures that $T_n(\boldx)=(1/n) \sum_{i=1}^n(m_i(\boldx)-\bar{m}(\boldx))^2$ is bounded below and above in probability, an essential requirement towards the proof of contraction rates of the variational posterior.
	
	Let $\theta=(\beta,B)$ be the model parameter and $\theta_0=(\beta_0,B_0)$ be the true parameter from which the data $X^{(n)}$ is generated. Let $L(\theta)$ and $L(\theta_0)$ denote the pseudo-likelihood as in \eqref{pseudo} under the model parameter and true parameter respectively. Further, let $L_0$ denote the true probability mass function from which $X^{(n)}$ is generated. Thus, $L_0$ is as in \eqref{likel} with $\theta=\theta_0$. We shall  use the notations  $\E_0^{(n)}$ and  $\mathbb{P}_0^{(n)}$ to denote expectation and probability mass function with respect to $L_0$.
	
	We next present the main theorem which establishes the  contraction rate for the variational posterior. Following the proof, we next establish the contraction rate of the variational Bayes estimator as  a corollary.  We shall use the term with dominating probability to imply that under $\mathbb{P}_0^{(n)}$, the probability of the event goes to 1 as $n \to \infty$.  
	
	\begin{theorem}[Posterior Contraction]
		\label{thm:post-convergence}
		Let $\mathcal{U}_{\varepsilon_n}=\{\theta:||\theta-\theta_0||_2\leq \varepsilon_n\}$ be neighborhood of the true parameters. Suppose $\epsilon_n$ satisfies assumption 2., then in $\mathbb{P}_0^{(n)}$ probability
		$$Q^*(\mathcal{U}_{\varepsilon_n}^c) \to 0, \: \: n \to \infty,$$
		where $\varepsilon_n=\epsilon_n\sqrt{M_n \log n}$ for any slowly increasing sequence $M_n \to \infty$ satisfying $\varepsilon_n \to 0$.
	\end{theorem}
	The above result establishes that the posterior distribution of $\beta$ and $B$ concentrates around the true value $\beta_0$ and $B_0$ at a rate slight larger than $\epsilon_n$. The proof of the above theorem rests on following lemmas, whose proofs have been deferred to the appendix \ref{append-A}.
	\begin{lemma}
		\label{lem:E1}
		There exists a constant $C_0>0$, such that for any $\epsilon_n \rightarrow 0$, $n\epsilon_n^2 \rightarrow \infty$,  
		\begin{align*}
			\mathbb{P}_0^{(n)} \left( \log  \int_{\mathcal{U}_{\epsilon_n}^c} \frac{L(\theta)}{L(\theta_0)}p(\theta)d\theta  \leq  -C_0 n\epsilon_n^2 \right) \rightarrow 1, \:\: n \to \infty.
		\end{align*}
	\end{lemma}
	\begin{lemma}
		\label{lem:E2}
		Let $\epsilon_n$ be the sequence satisfying the Assumption 2, then for any $C>0$,
		\begin{align*}
			\mathbb{P}_0^{(n)} \left(\Big|\log \int \frac{L(\theta)}{L(\theta_0)}p(\theta)d\theta \Big| \leq C n\epsilon_n^2 \log n \right) \rightarrow 1.
		\end{align*}
	\end{lemma}
	\begin{lemma}
		\label{lem:E3}
		Let $\epsilon_n$ be the sequence satisfying Assumption 2, then for  some $Q \in \mathcal{Q}^{\bf MF}$ and any $C>0$,
		\begin{align*}
			\mathbb{P}_0^{(n)} \left(\int \log  \frac{L(\theta_0)}{L(\theta)}q(\theta)d\theta \leq C n\epsilon_n^2 \log n \right) \rightarrow 1.
		\end{align*}
	\end{lemma}
	Lemma \ref{lem:E1} and Lemma \ref{lem:E2} taken together suffice to establish the posterior consistency of the true posterior based on the likelihood $L(\theta)$ as in \eqref{true_po}. Lemma \ref{lem:E3} on the other hand is the additional condition which needs to hold to ensure the consistency of the variational posterior. We next state an important result which relates the variational posterior to the true posterior.
	
	\noindent {\bf Formula for KL divergence:} By Corollary 4.15 in \cite{boucheron2013concentration}, 
	$${\rm KL}(P_1, P_2) =\sup_f \left[\int fdP_1-\log \int e^{f} dP_2\right].$$
	Using the above formula in the context of variational distributions, we get 
	\begin{equation}
		\label{e:kl-rel}
		\int fdQ^* \leq {\rm KL}(Q^*, \Pi(\mid X^{(n)}))+\log \int e^{f}d\Pi(\mid X^{(n)}).
	\end{equation}
	The above relation serves as an important tool towards the proof of Theorem \ref{thm:post-convergence}. 
	Next, we provide a brief sketch of the proof.  Further details on the proof have been deferred to appendix \ref{append-B}.
	
	\noindent {\bf Sketch of proof of Theorem \ref{thm:post-convergence}:} 
	
	\noindent  Let $f= (C_0/2)n  \varepsilon_n^2   1[\theta \in \mathcal{U}_{\varepsilon_n}^c]$, then
	\begin{align*}
		& (C_0/2)n \varepsilon_n^2   Q^*( \mathcal{U}_{\varepsilon_n}^c) \leq {\rm KL}(Q^*, \Pi(\mid X^{(n)}))+\log (e^{(C_0/2)n\varepsilon_n^2} \Pi( \mathcal{U}_{\varepsilon_n}^c \mid X^{(n)})+\Pi( \mathcal{U}_{\varepsilon_n} \mid X^{(n)}))\\
		&\implies Q^*( \mathcal{U}_{{\varepsilon_n}}^c) \leq \frac{2}{C_0n\varepsilon_n^2  }{\rm KL}(Q^*, \Pi(\mid X^{(n)}))+ \frac{2}{C_0 n\varepsilon_n^2  }\log (1+e^{(C_0/2)n \varepsilon_n^2}\Pi( \mathcal{U}_{{\varepsilon_n}}^c \mid X^{(n)})).
	\end{align*}
	By Lemma \ref{lem:E2} and \ref{lem:E3}, it can be established with dominating probability for any $C>0$, as $n \to \infty$ 
	$${\rm KL}(Q^*, \Pi(\mid X^{(n)})) \leq Cn\epsilon_n^2 \log n.$$ 
	By Lemma \ref{lem:E1} and \ref{lem:E2}, it can be established with dominating probability, as $n \to \infty$
	\begin{align}
		\label{e:true-post-conv}
		\Pi( \mathcal{U}_{{\varepsilon_n}}^c \mid X^{(n)}) \leq e^{-C_1 n  \varepsilon_n^2},
	\end{align}
	for any $C_1>C_0/2$. Therefore, with dominating probability
	\begin{align}
		\label{e:var-post-conv}
		Q^*( \mathcal{U}_{{\varepsilon_n}}^c) \leq \frac{2C}{C_0 M_n}+\frac{2}{  C_0 n\varepsilon_n^2}\log(1+ e^{-(C_1-C_0/2)n\varepsilon_n^2 }) \sim \frac{2C}{C_0 M_n}+\frac{e^{-(C_1-C_0/2)n\varepsilon_n^2}}{  C_0 n\varepsilon_n^2} \to 0.
	\end{align}
	This completes the proof.
	
	Note that \eqref{e:true-post-conv} gives the statement for the contraction of the true posterior. Similarly the contraction rate for the variational posterior follows as a consequence of \eqref{e:var-post-conv}. An important difference to note is that $Q^*(\mathcal{U}_{\varepsilon_n}^c)$ goes to 0 at the rate $1/M_n$ in contrast to the faster rate $e^{-C_1n\varepsilon_n^2}$ for the true posterior.
	
	Note, Theorem \ref{thm:post-convergence} gives the contraction rate of the variational posterior. However, the convergence of the of variational Bayes estimator to the true values of $\beta_0$ and $B_0$ is not immediate. The following corollary gives the convergence rate for the variational Bayes estimate as long as assumptions \ref{assumption_rs}, \ref{assumption_mf} and \ref{assumption_vb} hold. 
	
	\begin{corollary}[Variational Bayes Estimator Convergence]
		\label{thm: bayes-convergence}
		Let  $\varepsilon_n$ be as in Theorem \ref{thm:post-convergence}, then in $\mathbb{P}_0^{(n)}$ probability,
		$$ \frac{1}{\varepsilon_n} \E_{Q^*}(||\theta-\theta_0||_2) \to 0, \:\: n \to \infty.$$
	\end{corollary}
	
	Next, we provide a brief sketch of the proof. Further details of the proof have been deferred to appendix \ref{append-B}.
	
	\noindent {\bf Sketch of proof of Corollary \ref{thm: bayes-convergence}:} Let $f=(C_2/2)n\varepsilon_n||\theta-\theta_0||_2$, then
	\begin{align*}
		(C_2/2)&n\varepsilon_n \int ||\theta-\theta_0||_2 dQ^*(\theta) \\
		&\leq {\rm KL}(Q^*,\Pi(|X^{(n)}))+\log (\int e^{C_2n \varepsilon_n||\theta-\theta_0||_2/2} d\Pi(\theta|X^{(n)})).
	\end{align*}
	By Lemma \ref{lem:E2} and \ref{lem:E3}, it can be established with dominating probability, for any $C>0$
	$${\rm KL}(Q^*,\Pi(|X^{(n)})) \leq Cn\epsilon_n^2 \log n.$$
	By Lemma \ref{lem:E1},  and \ref{lem:E2}, it can be established with dominating probability, for some $C_2>0$
	\begin{align}
		\label{e:true-bayes-post}
		\int e^{(C_2/2)n\varepsilon_n||\theta-\theta_0||_2} d\Pi(\theta|X^{(n)}) \leq \frac{1}{(C_2/2)n\varepsilon_n^2}e^{C n\epsilon_n^2 \log n }.
	\end{align}
	Therefore, with dominating probability
	\begin{align*}
		\int ||\theta-\theta_0||_2 dQ^*(\theta)  \leq \frac{2C\varepsilon_n}{C_2M_n }-\frac{2\log(C_2/2)}{C_2n\varepsilon_n}-\frac {2\varepsilon_n \log (n\varepsilon_n^2)}{C_2n\varepsilon_n^2}+\frac{2C\varepsilon_n}{C_2M_n} \leq \varepsilon_n o(1).
	\end{align*}
	This completes the proof.
	
	\eqref{e:true-bayes-post} follows as a consequence of convergence of the true posterior. An important thing to note that if $\varepsilon_n$ can be made arbitrarily close to $n^{-\delta}$ for $0<\delta<1/2$, it guarantees close to $\sqrt{n}$ convergence.
	
	\section{Simulation Results}
	\label{section4}
	\subsection{Generating observed data}
	\label{section_generating_data}
	For numerical implementation, we need a coupling matrix $A_n$ and an observed vector $\boldx$ from \eqref{likel}. First, for generating a random $d$-regular graph and its scaled adjacency matrix, we used a python package \verb+NetworkX+. Using the scaled adjacency matrix as our coupling matrix $A_n$, we facilitate $Metropolis$--$Hastings$ algorithm to generate an observed vector $\boldx$ with true parameters $(\beta_0, B_0)$ as follows:
	\begin{itemize}
		\item[0.] Define $H(\boldx) = \frac{\beta_0}{2} \boldx^\top A_n \boldx + B_0\sum_{i=1}^n x_i$ and start with a random binary vector $\boldx = (x_1, \dots, x_n)^\top$.
		\item[1.] Randomly choose a spin $x_i$, $i \in \{1,\dots,n\}$.
		\item[2.] Flip the chosen spin, i.e. $x_i = -x_i$, and calculate $\Delta H = H(\boldx_{new}) - H(\boldx_{old})$ due to this flip.
		\item[3.] The probability that we accept $\boldx_{new}$ is:
		\begin{align*}
			\mathbb{P}( \text{accept}~ \boldx_{new}) = 
			\begin{cases}
				1, & \text{if}\, \Delta H > 0, \\
				\exp \left( \Delta H \right), & \text{otherwise}.
			\end{cases}
		\end{align*}
		\item[4.] If rejected, put the spin back, i.e. $x_i = -x_i$. 
		\item[5.] Go to 1 until the maximum number of iterations ($L$) is reached.
		\item[6.] After $L = 1,000,000$ iterations, the last result is a sample $\boldx$ we use.
	\end{itemize}
	One can read \cite{izenman2021sampling} for more details of sampling from Ising model.

	\subsection{Performance Comparison}
	We compare the performance of the parameter estimation methods for two-parameter Ising model \eqref{likel} under various combinations of $(d,n)$ and $(\beta_0,B_0)$. $Mean~squared~error (MSE)$ is used as the measurement for assessing the performances. With the given coupling matrix $A_n$ for each scenario, we repeat following steps $R$ times:
	\begin{itemize}
		\item Generate an observed vector $\boldx$ from \eqref{likel} with true parameters $(\beta_0, B_0)$.
		\item Using the proposed BBVI algorithm with MF family or BN family, obtain the optimal variational distribution $q^*$ .
		\item Get $\hat{\theta} = (\hat{\beta}, \hat{B})^\top$ as the sample means based on the samples drawn from $q^*$.
	\end{itemize}  
	We use $S=200$ or $S=2000$ as the Monte Carlo sample size. Figure 1. in the Supplement Section A.2 describes ELBO convergence for the two different sample sizes with MF family and BN family. The figure indicates that the ELBO converges well with a moderate choice of $S$. Further, for faster convergence, one might choose BN family over mean-field family for variational distributions.

	After we get $R=100$ pairs of estimates $(\hat{\beta}_1, \hat{B}_1), \dots, (\hat{\beta}_R, \hat{B}_R)$, $MSE$ is calculated as:
	\begin{align}
		MSE = \frac{1}{R} \sum_{r=1}^R \left(\left(\hat{\beta}_r - \beta_0 \right)^2 + \big(\hat{B}_r - B_0 \big)^2\right). 
	\end{align}
	
	First, we choose moderate values of $\beta_0 = 0.2, 0.7$ with $B_0 = \pm 0.2$, $\pm 0.5, \pm 0.8$. For each pair of $(\beta_0, B_0)$ we take $d =10, 50$. The simulation results show that the $MSE$ values of our BBVI algorithm are smaller than  PMLE from \cite{Ghosal:2020} for most cases. The two numbers in each cell of Table~\ref{perf1}, \ref{perf2}, and \ref{perf3} represent $MSE$ values when $n=100$ and $n=500$ respectively. 
	
	\begin{table}[h]
		\def~{\hphantom{0}}
		\caption{Mean squared errors and computation times  for each pair of $(\beta_0, B_0)$ when $n=100$ (left numbers) and $n=500$ (right numbers)  given the degree of underlying graph $(d)$.}
		\resizebox{\columnwidth}{!}{%
			\begin{tabular}{cccccccc}
				Degree of  & Method & Monte Carlo  & $(0.2, 0.2)$ &  $(0.2, -0.2)$ & $(0.7, 0.2)$ & $(0.7, -0.2)$ & Convergence \\ graph ($d$) &  &  samples ($S$) &  &   &  &  & time (sec)  \\ \hline
				10 & PMLE & - & 0.116 / 0.051 & 0.090 / 0.022 & 0.512 / 0.074 & 0.555 / 0.083 & 1.5 / 3.3  \\
				& MF family & 200  &  0.121 / 0.060 &  0.105 / 0.030 & 0.180 / 0.073 & 0.326 / 0.083 & 30.4 / 60.0 \\
				&  & 2000 & 0.107 / 0.052 &  0.092 / 0.021 & 0.137 / 0.076 & 0.267 / 0.076 & 120.3 / 245.9 \\ 
				& BN family & 200 & 0.100 / 0.047 & 0.084 / 0.019 & 0.204 / 0.075 & 0.324 / 0.084 & 32.8 / 68.0 \\
				&  & 2000 & 0.095 / 0.045 &  0.077 / 0.016 & 0.202 / 0.071 & 0.305 / 0.073 &  122.9 / 251.2 \\ \hline
				50 & PMLE & -  & 0.555 / 0.101 & 0.414 / 0.163 & 1.158 / 0.333 & 0.938 / 0.386 & 1.5 / 3.6 \\
				& MF family & 200 & 0.187 / 0.079 &  0.168 / 0.161 & 0.065 / 0.183  & 0.075 / 0.182 & 30.7 / 60.2 \\
				&  & 2000 &  0.155 / 0.072 &  0.144 / 0.107 & 0.070 / 0.144 & 0.080 / 0.155 & 120.1 / 246.6 \\ 
				& BN family & 200 & 0.293 / 0.065  & 0.236 / 0.148 & 0.102 / 0.174 & 0.085 / 0.178 & 31.5 / 68.1 \\
				&  & 2000 & 0.213 / 0.090 & 0.178 / 0.143 & 0.081 / 0.176 & 0.063 / 0.185 & 123.1 / 250.8
		\end{tabular}}
		\label{perf1}
		PMLE, pseudo maximum likelihood estimate \cite{Ghosal:2020}; MF, mean-field; BN, bivariate normal.
	\end{table}
	
	\begin{table}[h]
		\def~{\hphantom{0}}
		\caption{Mean squared errors and computation times  for each pair of $(\beta_0, B_0)$ when $n=100$ (left numbers) and $n=500$ (right numbers)  given the degree of underlying graph $(d)$.}
		\resizebox{\columnwidth}{!}{%
			\begin{tabular}{cccccccc}
				Degree of  & Method & Monte Carlo & $(0.2, 0.5)$ &  $(0.2, -0.5)$ & $(0.7, 0.5)$ & $(0.7, -0.5)$ & Convergence \\ graph ($d$) &  &  samples ($S$) &  &   &  &  & time (sec)  \\ \hline
				10 & PMLE & -  & 0.272 / 0.053 &  0.126 / 0.067 & 1.048 / 0.232 & 1.240 / 0.261 & 1.5 / 3.2 \\
				& MF family & 200  & 0.194 / 0.055 &  0.116 / 0.072 & 0.130 / 0.150 & 0.141 / 0.146 & 30.1 / 60.1 \\
				&  & 2000 &  0.144 / 0.041 &  0.098 / 0.051 &  0.126 / 0.151 &  0.122 / 0.132 & 120.5 / 244.3\\ 
				& BN family & 200 & 0.231 / 0.042 & 0.136 / 0.046 & 0.172 / 0.144  & 0.239 / 0.136 & 31.9 / 67.5 \\
				&  & 2000 & 0.217 / 0.047 &  0.130 / 0.058 & 0.156 / 0.140 & 0.220 / 0.133 & 122.1 / 250.3 \\ \hline
				50 & PMLE & -  & 1.282 / 0.196 & 0.777 / 0.239  & 1.900 / 0.765 & 1.923 / 1.216 & 1.4 / 3.1\\
				& MF family & 200 &  0.144 / 0.113 &  0.128 / 0.234 &  0.089 / 0.162 & 0.086 / 0.254 & 30.3 / 60.0 \\
				&  & 2000 & 0.120 / 0.048 &  0.108 / 0.097 & 0.099 / 0.197 & 0.100 / 0.157 & 120.4 / 245.1 \\ 
				& BN family & 200 & 0.320 / 0.105 & 0.296 / 0.138 & 0.053 / 0.138 & 0.070 / 0.194 & 32.0 / 67.9\\
				&  & 2000 & 0.296 / 0.126 & 0.274 / 0.183 & 0.053 / 0.107 & 0.065 / 0.135 & 123.4 / 251.7
		\end{tabular}}
		\label{perf2}
		PMLE, pseudo maximum likelihood estimate \cite{Ghosal:2020}; MF, mean-field; BN, bivariate normal.
	\end{table}
	
	\begin{table}[h]
		\def~{\hphantom{0}}
		\caption{Mean squared errors and computation times  for each pair of $(\beta_0, B_0)$ when $n=100$ (left numbers) and $n=500$ (right numbers)  given the degree of underlying graph $(d)$.}
		\resizebox{\columnwidth}{!}{%
			\begin{tabular}{cccccccc}
				Degree of  & Method & Monte Carlo & $(0.2, 0.8)$ &  $(0.2, -0.8)$ & $(0.7, 0.8)$ & $(0.7, -0.8)$ & Convergence \\ graph ($d$) &  &  samples ($S$) &  &   &  &  & time (sec)  \\ \hline
				10 & PMLE & -  & 0.240 / 0.136 & 0.123 / 0.124 & 1.502 / 0.667 & 1.125 / 0.828 & 1.5 / 3.0 \\
				& MF family & 200  &  0.151 / 0.110 & 0.126 / 0.104 & 0.128 / 0.142 & 0.105 / 0.215  & 30.2 / 60.2 \\
				&  & 2000 &  0.127 / 0.080 & 0.111 / 0.067 & 0.083 / 0.121 & 0.079 / 0.171 & 120.0 / 242.9 \\ 
				& BN family & 200 & 0.274 / 0.120 & 0.210 / 0.106 & 0.337 / 0.211 & 0.244 / 0.303 & 32.0 / 67.1 \\
				&  & 2000 & 0.264 / 0.123 & 0.213 / 0.111 & 0.294 / 0.134 & 0.225 / 0.219 & 121.8 / 251.0 \\ \hline
				50 & PMLE & -  & 0.922 / 0.494 & 1.756 / 0.517 & 2.592 / 2.773 & 2.510 / 1.933 & 1.5 / 3.1 \\
				& MF family & 200 & 0.157 / 0.142 & 0.170 / 0.120 & 0.275 / 0.125 & 0.283 / 0.119 & 30.3 / 60.0 \\
				&  & 2000 &  0.138 / 0.049 &  0.140 / 0.044 & 0.288 / 0.138 & 0.292 / 0.156 & 120.5 / 242.5 \\ 
				& BN family & 200 & 0.431 / 0.160 & 0.496 / 0.223 & 0.318 / 0.130 & 0.358 / 0.086 & 32.1 / 67.0 \\
				&  & 2000 & 0.417 / 0.183 & 0.474 / 0.196 & 0.305 / 0.094 & 0.347 / 0.074 & 121.5 / 251.2
		\end{tabular}}
		\label{perf3}
		PMLE, pseudo maximum likelihood estimate \cite{Ghosal:2020}; MF, mean-field; BN, bivariate normal.
	\end{table}
	For moderate values of $\beta_0$ (Table \ref{perf1}, \ref{perf2}, and \ref{perf3}), there is no significant difference between MF family and BN family in terms of $MSE$. 
	When the interaction parameter $\beta_0$ is large, BN family seems to perform better. Table~\ref{perf4} shows all results for $\beta_0 = 1.2$. 
	
	\begin{table}[h]
		\def~{\hphantom{0}}
		\caption{Mean squared errors and computation times  for each pair of $(\beta_0, B_0)$ when $n=100$ (left numbers) and $n=500$ (right numbers)  given the degree of underlying graph $(d)$.}
		\resizebox{\columnwidth}{!}{%
			\begin{tabular}{cccccccc}
				Degree of  & Method & Monte Carlo & $(1.2, 0.2)$ &  $(1.2, -0.2)$ & $(1.2, 0.5)$ & $(1.2, -0.5)$ & Convergence \\ graph ($d$) &  &  samples ($S$) &  &   &  &  & time (sec)  \\ \hline
				10 & PMLE & -  &  1.681 / 0.379 & 1.677 / 0.496 & 2.687 / 1.483 & 2.716 / 1.598 & 1.6 / 3.0 \\
				& MF family & 200  &  0.544 / 0.315 & 0.578 / 0.428 & 0.501 / 0.627 & 0.498 / 0.737 & 30.1 / 60.2 \\
				&  & 2000 &  0.573 / 0.417 & 0.625 / 0.534 & 0.532 / 0.700 & 0.514 / 0.814 & 120.8 / 241.8 \\ 
				& BN family & 200 &  0.404 / 0.309 & 0.409 / 0.374 & 0.247 / 0.488 & 0.248 / 0.479 & 32.2 / 66.9 \\
				&  & 2000 &  0.396 / 0.295 & 0.405 / 0.368 & 0.235 / 0.411 & 0.234 / 0.449 & 121.3 / 250.7
				\\ \hline
				50 & PMLE & - &  2.941 / 2.302 & 3.252 / 1.509 & 2.830 / 3.190 & 5.631 / 3.526 & 1.5 / 3.1 \\
				& MF family & 200 &  0.775 / 0.998 & 0.771 / 0.775 & 0.832 / 0.836 & 0.810 / 0.792 & 30.0 / 60.3 \\
				&  & 2000 &  0.812 / 1.022 & 0.818 / 1.025 & 0.868 / 0.972 & 0.862 / 0.947 & 120.5 / 241.4 \\ 
				& BN family &  200 & 0.308 / 0.580 & 0.286 / 0.518 & 0.376 / 0.336 & 0.369 / 0.294 & 31.9 / 67.0 \\
				&  & 2000 &  0.321 / 0.644 & 0.285 / 0.496 & 0.380 / 0.272 & 0.376 / 0.208 & 121.0 / 250.1\\
		\end{tabular}}
		\label{perf4}
		PMLE, pseudo maximum likelihood estimate \cite{Ghosal:2020}; MF, mean-field; BN, bivariate normal.
	\end{table}
	The numerical studies validate the superiority of our proposed variational Bayes based method. For more practical applications, we used our algorithm to regenerate an image matrix in the next subsection.
	
	\subsection{Data Reconstruction}
	The Ising model can be used for constructing an image in computer vision problem. In particular, the Bayesian procedure facilitate the reconstruction easily by using the posterior predictive distribution \cite{halim2007modified}. Consider an image in which each pixel represents either $-1 (white)$ or $1 (black)$. For choice of coupling matrix $A_n$, we use four-nearest neighbor structure and construct corresponding scaled adjacency matrix \cite{hurn2003tutorial}. Then, we can generate such images following the steps in the subsection \ref{section_generating_data} with a true parameter pair $(\beta_0, B_0)$ and use it as our given data $\boldx$. With the generated image $\boldx$ and coupling matrix $A_n$, we obtain $(\hat{\beta}, \hat{B})$ after implementing the parameter estimation procedure based on BN family. The estimates $(\hat{\beta}, \hat{B})$ are used for data regeneration following the steps in the subsection \ref{section_generating_data} again. In Figure~\ref{img_recon}, we plot two original images in left column. The first original image was generated with $\beta_0 = 1.2, B_0 = 0.2$ and we use $\beta_0 = 1.2, B_0 = -0.2$ for the second one. Also, in the right column, there are two corresponding images regenerated. It seems that, using our two-parameter Ising model and VI method with BN family, we can reconstruct the overall tendency of black and white images fairly well. For more precise pixel-by-pixel reconstruction, one can utilize multiple external parameters, $\boldB = (B_1, \dots, B_n)^\top$.
	
	\begin{figure}[htbp]
		\centering
		\begin{subfigure}{.5\textwidth}
			\centering
			\includegraphics[width=.7\linewidth]{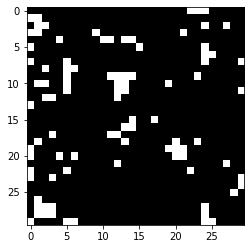}
		\end{subfigure}%
		\begin{subfigure}{.5\textwidth}
			\centering
			\includegraphics[width=.7\linewidth]{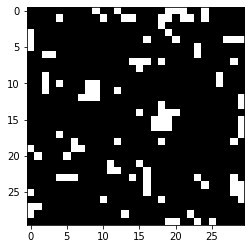}
		\end{subfigure}
		\centering
		\begin{subfigure}{.5\textwidth}
			\centering
			\includegraphics[width=.7\linewidth]{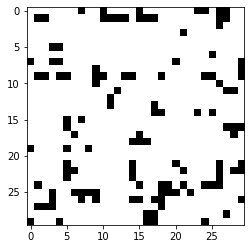}
		\end{subfigure}%
		\begin{subfigure}{.5\textwidth}
			\centering
			\includegraphics[width=.7\linewidth]{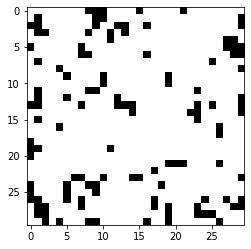}
		\end{subfigure}
		\caption{Figures showing that ELBO convergence.}
		\label{img_recon}
	\end{figure}
	
	\section{Conclusion}
	In this article, we consider a two-parameter Ising model and proposed a variational Bayes estimation technique. The use of pseudo-likelihood helps to avoid the computation of normalizing constants. The VI procedure facilitates the computation further. We established the theoretical properties of the proposed procedure, which is quite challenging yet needed for statistical validation. The numerical investigation indicates the procedure can work well in applications. 
	There are variations of the Ising models depending on applications. The framework developed in this work will facilitate further developments of this important statistical tool.

	\begin{appendix}
		\section{}\label{append-A}
		\subsection{Preliminary notations and Lemmas}
		
		\noindent Let $\theta=(\beta,B)$. Define $W_n=(W_{1n},W_{2n})$
		\begin{align}
			\label{e:first-order}
			\nonumber    W_{1n}(\theta \mid \boldx) &= \sum_{i=1}^n m_i(\boldx) \left( x_i - \tanh (\beta m_i(\boldx) +B) \right), \\
			W_{2n}(\theta \mid \boldx) &= \sum_{i=1}^n \left( x_i - \tanh (\beta m_i(\boldx) +B) \right), \end{align}
		Define     \begin{align}
			\label{e:second-order}
			H_n(\theta \mid \boldx) &= \begin{bmatrix}
				\sum_{i=1}^n m_i(\boldx)^2 S_i(\theta \mid \boldx) & \sum_{i=1}^n m_i(\boldx)S_i(\theta \mid \boldx) \vspace{2mm}\\
				\sum_{i=1}^n m_i(\boldx)S_i(\theta \mid \boldx) & \sum_{i=1}^n S_i(\theta \mid \boldx)
			\end{bmatrix}
		\end{align}
		where  $S_i(\theta \mid \boldx) = \text{sech}^2 \left(\beta m_i(\boldx) +B \right)$.
		
		\noindent Define
		
		\begin{align}
			\label{e:Cn-def}
			R_{1n}(\theta \mid \boldx) &:= \begin{bmatrix}
				\sum_{i=1}^n m_i(\boldx)^3 \left(h_i(\theta \mid \boldx) - h_i^3( \theta \mid \boldx)\right) & \sum_{i=1}^n m_i(\boldx)^2 \left(h_i(\theta \mid \boldx) - h_i^3( \theta \mid \boldx)\right)\vspace{2mm}\\
				\sum_{i=1}^n m_i(\boldx)^2 \left(h_i(\theta \mid \boldx) - h_i^3( \theta \mid \boldx)\right) & \sum_{i=1}^n m_i(\boldx) \left(h_i(\theta \mid \boldx) - h_i^3( \theta \mid \boldx)\right)
			\end{bmatrix},
		\end{align}
		and 
		\begin{align}
			\label{e:Dn-def}
			R_{2n}(\theta \mid  \boldx) &:= \begin{bmatrix}
				\sum_{i=1}^n m_i(\boldx)^2 \left(h_i(\theta \mid \boldx) - h_i^3( \theta \mid \boldx)\right) & \sum_{i=1}^n m_i(\boldx) \left(h_i(\theta \mid \boldx) - h_i^3( \theta \mid \boldx)\right)\vspace{2mm}\\
				\sum_{i=1}^n m_i(\boldx) \left(h_i(\theta \mid \boldx) - h_i^3( \theta \mid \boldsymbol{x})\right) & \sum_{i=1}^n \left(h_i(\theta \mid \boldx) - h_i^3( \theta \mid \boldx)\right)
			\end{bmatrix}.
		\end{align}
		where $h_i(\theta \mid \boldx) = \tanh\left( \beta m_i(\boldx) + B\right)$.
		
		\begin{lemma}
			\label{qn:rn:bounded} 
			Let $W_{1n}$ and $W_{2n}$ be as in \eqref{e:first-order}, then
			\begin{align*}
				\frac{1}{n}\E_0^{(n)}\left( W_{1n}(\theta_0 \mid \boldx)\right)^2 < \infty \hspace{10mm}
				\frac{1}{n}\E_0^{(n)}\left( W_{2n}(\theta_0 \mid \boldx)\right)^2 < \infty
			\end{align*}
			\begin{proof}
				See the lemma 2.1 in \cite{Ghosal:2020}.
			\end{proof}
		\end{lemma}
		
		\begin{lemma}
			\label{Lee}
			Let $p_1$ and $p_2$ be any two density functions. Then,
			\begin{align*}
				\E_{P_1}\left(\left|\log \frac{p_1}{p_2} \right|\right) \leq {\rm KL}(P_1, P_2) + \frac{2}{e}
			\end{align*}
			\begin{proof}
				See the lemma 4 in \cite{lee2000consistency}.
			\end{proof}
		\end{lemma}
		
		\begin{lemma}
			\label{lem:}  
			Let $T_n(\boldx) = \frac{1}{n}\sum_{i=1}^n \left( m_i(\boldx) - \frac{1}{n}\sum_{i=1}^n m_i(\boldx)\right)^2$. Suppose Assumptions \ref{assumption_rs}, \ref{assumption_mf} and \ref{assumption_vb}  hold, then
			$$T_n(\boldx)=O_p(1), \hspace{5mm} 1/T_n(\boldx)=O_p(1)$$
			\begin{proof}
				See the theorem 1.4 in \cite{Ghosal:2020}.
			\end{proof}
		\end{lemma}
		
		\subsection{Taylor expansion for $\log L(\theta)$}
		
		\begin{lemma}
			\label{hessian}
			Consider the term $(\theta - \theta_0)^\top H_n(\theta_0 \mid \boldx) (\theta - \theta_0)$ where $H_n$ is the same as in \eqref{e:second-order}. Then, for some $C_1,C_2>0$, we have
			\begin{align*}
				\mathbb{P}_0^{(n)}\left(   C_1 n\lVert\theta - \theta_0\rVert_2^2 \leq (\theta - \theta_0)^\top H_n(\theta_0 \mid \boldx)(\theta - \theta_0) \leq C_2n \lVert\theta - \theta_0\rVert_2^2 \right) \to 1, n \to \infty 
			\end{align*}
			
			\begin{proof}
				For some $M_1,M_2>0$, let $\mathcal{A}_{1n}=\{\boldx: T_n(\boldx) \leq M_1\}$, $\mathcal{A}_{2n}=\{\boldx: T_n(\boldx)\geq M_2\}$. 
				
				\noindent Let $\mathcal{A}_n =\mathcal{A}_{1n} \cap \mathcal{A}_{2n}$, then  $\mathbb{P}_0^{(n)}(\mathcal{A}_n) \to 1$. 
				
				\noindent This is because  by Lemma \ref{lem:}, there exists $M_1$ and $M_2$ such that 
				\begin{align*}\mathbb{P}_0^{(n)}(T_n>M_1)&=\mathbb{P}_0^{(n)}(1/T_n<1/M_1)\to 0\\
					\mathbb{P}_0^{(n)}(T_n<M_2)&=\mathbb{P}_0^{(n)}(1/T_n>1/M_2) \to 0
				\end{align*}
				
				\noindent The remaining part of the proof works with only $\boldx \in \mathcal{A}_n$. Let $e^H_1 \geq e^H_2$ be the eigenvalues of $H_n(\theta_0 \mid \boldx)$. The trace of $H_n(\theta_0 \mid \boldx)$, is
				\begin{align*}
					tr\left(H_n(\theta_0 \mid \boldx)\right) &= e^H_1 + e^H_2 = \sum_{i=1}^n \text{sech}^2\left(\beta_0m_i(\boldsymbol{x}) + B_0 \right)\left(m_i^2(\boldsymbol{x}) + 1 \right)  \leq n(1+\gamma^2) \end{align*}
				where we used $|m_i(\boldx)| \leq \gamma$ based on Assumption \ref{assumption_rs}.
				Note (2.7) in \cite{Ghosal:2020} gives a lower bound of $e^H_2$:
				\begin{align}
					\label{eigen_lowerb}
					e^H_2 \geq \frac{\text{sech}^4(\beta_0 \gamma + |B_0|)}{1+\gamma^2}nT_n(\boldx),
				\end{align}
				where $T_n(\boldx)$ is as in the Lemma \ref{lem:}. By spectral decomposition of $H_n(\theta_0 \mid \boldx)$, 
				\begin{align}
					\label{e:hess-upp}
					\nonumber    (\theta - \theta_0)^\top H_n(\theta_0 \mid \boldx)(\theta - \theta_0) &\leq e^H_1\left\{(\beta - \beta_0)^2 + (B - B_0)^2\right\} \\
					\nonumber   &  \hspace{-15mm}\leq \left(n(1+\gamma^2) - e^H_2 \right) \left\{(\beta - \beta_0)^2 + (B - B_0)^2\right\} \\
					&  \hspace{-15mm} \leq n\left((1+\gamma^2) - \frac{\text{sech}^4(\beta_0 \gamma + |B_0|)}{1+\gamma^2}T_n(\boldx) \right) \left\{(\beta - \beta_0)^2 + (B - B_0)^2\right\}
				\end{align}
				Also, 
				\begin{align}
					\label{e:h-pos}
					\nonumber   (\theta - \theta_0)^\top H_n(\theta_0 \mid \boldx)(\theta - \theta_0) &\geq e^H_2\left\{(\beta - \beta_0)^2 + (B - B_0)^2\right\} \\
					&\geq  \frac{\text{sech}^4(\beta_0 \gamma + |B_0|)}{1+\gamma^2}nT_n(\boldx)  \left\{(\beta - \beta_0)^2 + (B - B_0)^2\right\}
				\end{align}
				Since $M_2\leq T_n(\boldx)\leq M_1$ for every $\boldx \in \mathcal{A}_n$, the proof follows.
			\end{proof}
		\end{lemma}
		
		\begin{lemma}
			\label{remainder}
			For $R_{1n}$ and $R_{2n}$ as in \eqref{e:Cn-def} and \eqref{e:Dn-def} respectively, let
			\begin{align}
				\label{e:third-order}
				&3 R_n(\tilde{\theta},\theta-\theta_0 \mid \boldx)\\
				& =   (\beta-\beta_0)(\theta -\theta_0)^{\top} R_{1n}(\tilde{\theta} \mid \boldx) (\theta - \theta_0) + (B - B_0)(\theta - \theta_0)^{\top} R_{2n}(\tilde{\theta} \mid  \boldx) (\theta - \theta_0)
			\end{align}
			where $R_n=R_n(\tilde{\theta},\theta-\theta_0 \mid \boldx)$ and $\tilde{\theta}=\theta_0+c(\theta-\theta_0)$ $0<c<1$. 
			
			\noindent Then, as $n \to \infty $ for some $C_1,C_2>0$ we have
			\begin{align*}
				\mathbb{P}_0^{(n)}  (  M_1 n \Delta^*  \leq R_n \leq M_2 n \Delta^*) \to 1, \
			\end{align*}
			where $\Delta^*= ( (\beta-\beta_0)\gamma + (B-B_0) )   \lVert\theta_0 - \theta\rVert_2^2$
			
			\begin{proof}
				For some $M_1,M_2>0$, let $\mathcal{A}_{1n}=\{\boldx: T_n(\boldx) \leq M_1\}$, $\mathcal{A}_{2n}=\{\boldx: T_n(\boldx)\geq M_2\}$.
				
				\noindent Let $\mathcal{A}_n =\mathcal{A}_{1n} \cap \mathcal{A}_{2n}$, then  $\mathbb{P}_0^{(n)}(\mathcal{A}_n) \to 1$. 
				
				\noindent This is because by Lemma \ref{lem:}, there exists $M_1,M_2>0$ such that
				\begin{align*}
					\mathbb{P}_0^{(n)}(T_n>M_1)&=\mathbb{P}_0^{(n)}(1/T_n<1/M_1) \to 0\\
					\mathbb{P}_0^{(n)}(T_n<M_2)&=\mathbb{P}_0^{(n)}(1/T_n>1/M_2) \to 0
				\end{align*}
				The remaining part of the proof works with only $\boldx \in \mathcal{A}_n$. 
				
				\noindent The determinant of $R_{1n}(\tilde{\theta} \mid  \boldx)$ is:
				\begin{align*}
					&det(R_{1n}(\tilde{\theta} \mid  \boldx))\\ &= \frac{1}{2} \sum_{i,j=1}^n m_i(\boldx) m_j(\boldx)  \left( h_i(\boldx) - h_i^3(\boldx) \right) \left( h_j(\boldx)- h_j^3(\boldx) \right)\left( m_i(\boldx) -  m_j(\boldx)\right)^2
				\end{align*}
				where $h_i(\boldx) = \tanh( \tilde{\beta} m_i(\boldx) + \tilde{B})$. Since $\tanh(\cdot) - \tanh^3(\cdot)$ has maximum value $0.38$ at $\sqrt{3}/{3}$ and $m_i(\boldx) \leq \gamma$ by Assumption \ref{assumption_rs}. Therefore, 
				\begin{align*}
					|  det(R_{1n}(\tilde{\theta}\mid \boldx))| &\leq \frac{1}{2} \gamma^2 (0.38)^2 \sum_{i=1}^n \sum_{j=1}^n \left( m_i(\boldx) - m_j(\boldx) \right)^2  = \gamma^2 (0.38)^2 n^2 T_n(\boldx).
				\end{align*}
				The trace of $R_{1n}(\tilde{\theta}\mid \boldx)$ is:
				\begin{align*}
					tr(R_{1n}(\tilde{\theta}\mid \boldx)) &= \sum_{i=1}^n m_i(\boldx) (h_i(\boldx) - h_i^3(\boldx) )\left(m_i^2(\boldx) + 1\right) \leq n 0.38 \gamma (1 + \gamma^2).
				\end{align*}
				Let $e_1^{R_{1n}} \geq e_2^{R_{1n}}$ be eigenvalues of $R_{1n}(\tilde{\theta}\mid \boldx)$.
				\begin{align*}
					e_2^{R_{1n}} &\geq \frac{e_1^{R_{1n}}e_2^{R_{1n}}}{e_1^{R_{1n}} + e_2^{R_{1n}}} = \frac{det(R_{1n}(\tilde{\theta}\mid \boldx))}{tr (R_{1n}(\tilde{\theta}\mid \boldx))} &\geq -\frac{\gamma^2 (0.38)^2 n^2 T_n(\boldx)}{n 0.38 \gamma (1 + \gamma^2)} = -\frac{0.38\gamma}{ 1 + \gamma^2} nT_n(\boldx).
				\end{align*}
				Therefore,
				\begin{align}
					\label{Cn_lower_b}
					(\theta - \theta_0)^{\top} R_{1n}(\tilde{\theta}\mid \boldx) (\theta - \theta_0) 
					&\geq e_2^{R_{1n}}||\theta-\theta_0||_2^2\geq -\frac{0.38\gamma}{ 1 + \gamma^2} nT_n(\boldx)||\theta-\theta_0||_2^2
				\end{align}
				and
				\begin{align}
					\label{Cn_upper_b}
					(\theta - \theta_0)^\top R_{1n}(\tilde{\theta}\mid \boldx) (\theta - \theta_0) 
					&\leq e_1^{R_{1n}}||\theta-\theta_0||_2^2 = (tr(R_{1n}(\tilde{\theta}\mid \boldx)) - e_2^{R_{1n}}))||\theta-\theta_0||_2^2 \nonumber\\
					&\leq 0.38\gamma n \left( (1+\gamma^2) +  \frac{T_n(\boldx)}{ 1 + \gamma^2}\right) ||\theta-\theta_0||_2^2.
				\end{align}
				With the same argument, we can get:
				\begin{align}
					\label{Dn_lower_b}
					(\theta - \theta_0)^{\top} R_{2n}(\tilde{\theta}\mid \boldx) (\theta - \theta_0) 
					&\geq -\frac{0.38}{ 1 + \gamma^2} nT_n(\boldx)||\theta-\theta_0||_2^2,
				\end{align}
				\begin{align}
					\label{Dn_upper_b}
					(\theta - \theta_0)^t R_{2n}(\tilde{\theta}\mid \boldx) (\theta - \theta_0) 
					&\leq 0.38 n \left( (1+\gamma^2) +  \frac{T_n(\boldx)}{ 1 + \gamma^2}\right) ||\theta-\theta_0||_2^2
				\end{align}
				Using \eqref{Cn_lower_b}, \eqref{Cn_upper_b}, \eqref{Dn_lower_b} and \eqref{Dn_upper_b} and noting  $M_2\leq T_n(\boldx)\leq M_1$ for every $\boldx \in \mathcal{A}_n$, the proof follows.
			\end{proof}
		\end{lemma}
		
		\begin{lemma}
			\label{q_p_relation}
			Let $q(\theta) \in \mathcal{Q}^{MF}$ with $\mu_1 = \log\beta_0$, $\mu_2 = B_0$, and $\sigma_1^2 = \sigma_2^2 = 1/n$, then 
			\begin{align*}
				\frac{1}{n\epsilon_n^2 \log n}{\rm KL}(Q,P)\to 0, \:\: n\to \infty
			\end{align*}
		\end{lemma}
		
		\begin{proof}
			Using the same notation in \eqref{mf:family}, the first term is:
			\begin{align*}
				{\rm KL}(Q,P) &= \E_{q_\beta(\beta)q_B(B)}\left( \log q_\beta(\beta) + \log q_B(B) - \log p_\beta(\beta) - \log p_B(B) \right) \\
				&= {\rm KL}\left(Q_\beta, P_\beta\right) + {\rm KL}\left(Q_B, P_B\right)\\
				&= \frac{1}{2}\left((\log \beta_0)^2 + \frac{1}{n} + B_0^2  + \frac{1}{n} - 2\right) + \log n =o(n\epsilon_n^2 \log n ),\:\:\:\: \text{since $n\epsilon_n^2 \to \infty$}
		\end{align*}\end{proof}
		
		\subsection{Technical details of Lemma \ref{lem:E1}}
		
		\begin{proof} [Proof of Lemma \ref{lem:E1}] Let $\mathcal{V}_{\epsilon_n}=\{|\beta-\beta_0|<\epsilon_n,|B-B_0|<\epsilon_n\}$. Then $\mathcal{V}_{\epsilon_n}\subseteq \mathcal{U}_{\sqrt{2}\epsilon_n}$ which implies $\mathcal{U}_{\sqrt{2}\epsilon_n}^c \subseteq
			\mathcal{V}_{\epsilon_n}^c$ which further implies
			\begin{equation}
				\label{e:u-v}
				\log  \int_{\mathcal{U}_{\sqrt{2} \epsilon_n}^c} \frac{L(\theta)}{L(\theta_0)}p(\theta)d\theta \leq \log\int_{\mathcal{V}_{ \epsilon_n}^c} \frac{L(\theta)}{L(\theta_0)}p(\theta)d\theta 
			\end{equation}
			We shall now establish for some $C_0>0$
			
			\begin{align*}
				\mathbb{P}_0^{(n)}\left(\log \int_{\mathcal{V}_{ \epsilon_n}^c} \frac{L(\theta)}{L(\theta_0)}p(\theta)d\theta \leq -C_0
				n\epsilon_n^2\right) &\to 1, \:\: n \to \infty,
			\end{align*}
			which in lieu of \eqref{e:u-v} completes the proof. 
			
			\noindent Define $\mathcal{A}_{1n}=\{\boldx: W_{1n}(\theta_0\mid \boldx)^2+W_{1n}(\theta_0\mid \boldx)^2 \leq n^{2/3}\}$ and $\mathcal{A}_{2n}=\{\boldx: T_n(\boldx)\geq M\}$ for some $M>0$. Define $\mathcal{A}_n =\mathcal{A}_{1n} \cap \mathcal{A}_{2n}$. 
			
			\noindent Here $\mathbb{P}_0^{(n)}(\mathcal{A}_n) \to 1$. This because by Markov's inequality and Lemma \ref{qn:rn:bounded},
			\begin{align*}
				&\mathbb{P}_0^{(n)}\left(W_{1n}(\theta_0\mid \boldx)^2 + W_{2n}(\theta_0\mid \boldx)^2 > n^{2/3}\varepsilon\right)\\
				&\leq \frac{1}{n^{4/3}}\E_0^{(n)} \left(W_{1n}(\theta_0\mid \boldx)^2 + W_{2n}(\theta_0\mid \boldx)^2  \right)  \to 0.
			\end{align*}
			and by Lemma \ref{lem:} $\mathbb{P}_0^{(n)}(T_n<M)=\mathbb{P}_0^{(n)}(1/T_n>1/M) \to 0$.
			
			\noindent We shall show for $\boldx \in \mathcal{A}_n$, $L(\theta)/L(\theta_0)\leq e^{-C_0n\epsilon_n^2}$, $\forall$ $\theta\in \mathcal{V}_{\epsilon_n}^c$ which implies $\forall $ $\boldx \in \mathcal{A}_n$,
			
			\begin{align}
				\label{e:v-upper}
				\nonumber \log \int_{\mathcal{V}_{\epsilon_n}^c} (L(\theta)/L(\theta_0))p(\theta)d\theta &=
				\log \int_{\mathcal{V}_{\epsilon_n}^c} (L(\theta)/L(\theta_0))p(\theta)d\theta\\
				&\leq \log (e^{-C_0n\epsilon_n^2}\int_{\mathcal{V}_{\epsilon_n}^c} p(\theta)d\theta)\leq -C_0n\epsilon_n^2
			\end{align}
			since $p(\mathcal{V}_{\epsilon_n^c})\leq 1$. This completes the proof since $\mathbb{P}_0^{(n)}(\mathcal{A}_n) \to 1$ as $n\to \infty$.
			
			Next, note that $\mathcal{V}_{\epsilon_n}^c$ is given by the union of the following terms\begin{eqnarray*}
				V_{1n}&=\{(\beta,B): \beta-\beta_0\geq \epsilon_n , B\geq B_0\}, 
				V_{2n}=\{(\beta,B): \beta-\beta_0\geq \epsilon_n , B < B_0\}\\ 
				V_{3n}&=\{  (\beta,B): \beta-\beta_0 <  -\epsilon_n , B \geq B_0\},
				V_{4n}=\{  (\beta,B): \beta-\beta_0 <  -\epsilon_n , B < B_0\} \\
				V_{5n}&=\{  (\beta,B): \beta\geq \beta_0, B -B_0 \geq \epsilon_n \},
				V_{6n}=\{  (\beta,B): \beta < \beta_0, B -B_0 \geq \epsilon_n \} \\
				V_{7n}&=\{  (\beta,B): \beta\geq \beta_0, B -B_0 < -\epsilon_n \},
				V_{8n}=\{  (\beta,B): \beta < \beta_0, B -B_0 < -\epsilon_n \} 
			\end{eqnarray*}
			We shall now show for $\boldx \in \mathcal{A}_n$ and $\theta \in V_{1n}$, $L(\theta)/L(\theta_0)\leq e^{-C_0n\epsilon_n^2}$. The proof of other parts follow similarly. 
			
			\noindent (a) Let $\theta = (\beta, B)$ and $\theta_0' = (\beta_0+\epsilon,B_0)$, where $\beta\geq \beta_0+\epsilon$ and $B\geq B_0$. Also, define $$\theta_t=\theta_0'+t(\theta-\theta_0')~ \text{where}~ 0<t<1.$$
			Consider a function $g$:
			$$g(t)=f(\theta_t)=\log L(\theta_t)-\log L(\theta_0')-\Delta_n(\theta_0')^\top (\theta_t-\theta_0'),$$
			where $\Delta_n(\theta) = \left(\nabla_\beta \log L(\theta), \nabla_B \log L(\theta)  \right)^\top$. Note that $g(t)$ is a function of $t$. We want to show $g(t)\leq g(0)$ provided $t>0$. We shall instead show $g'(t)\leq 0$. By Taylor expansion,
			$$g'(t)=g'(0)+ g''(\tilde{t}) t.$$ 
			for some  $\tilde{t}\in[0,t]$. Here, $g'(0)=0$ and $g''(\tilde{t})=-(\theta-\theta_0')^\top H_n(\theta_{\tilde{t}}\mid \boldx) (\theta-\theta_0') \leq 0$
			where $H_n$ as in \eqref{e:second-order} is a positive definite matrix (by \eqref{e:h-pos} in \ref{hessian} and $T_n(\boldx)\geq 0$).  Since $g(t)$ is decreasing for $0<t<1$, thus $$g(1)\leq g(0) \implies f(\theta) \leq  f(\theta_0')$$
			
			\noindent(b) Similarly, let $\theta=(\beta,B)$ $\theta_0''=(\beta_0,B_0+\epsilon)$, where $\beta\geq \beta_0$ and $B\geq B_0+\epsilon$. Define $$\theta_t=\theta_0'+t(\theta-\theta_0'')~ \text{where}~ 0<t<1.$$
			$$h(t)=f(\theta_t)=\log L(\theta_t)-\log L(\theta_0'')-\Delta_n(\theta_0'')^\top(\theta_t-\theta_0'').$$
			With similar argument in (a), we conclude that
			$h(1)\leq h(0) \implies f(\theta) \leq f(\theta_0'')$. Therefore,
			\begin{align}
				\nonumber&\sup_{\theta \in V_{1n}} (\log L(\theta)-\log L(\theta_0))\\
				\nonumber&\leq \sup_{ \{\beta -\beta_0\in [\epsilon_n,\epsilon], B\geq B_0\}}
				(\log L(\theta)-\log L(\theta_0))+\sup_{ \{\beta > \beta_0+\epsilon, B\geq B_0\}}
				(\log L(\theta)-\log L(\theta_0))\\
				\nonumber&\leq \sup_{ \{\beta -\beta_0\in [\epsilon_n,\epsilon], B\geq B_0\}}
				(\log L(\theta)-\log L(\theta_0))+(\log L(\theta_0')-\log L(\theta_0))\\
				\nonumber &\leq \sup_{ \{\beta -\beta_0\in [\epsilon_n,\epsilon], B-B_0 \in [0,\epsilon]\}}
				(\log L(\theta)-\log L(\theta_0))+\sup_{ \{\beta-\beta_0\in[\epsilon_n,\epsilon], B >B_0+\epsilon\}}(\log L(\theta)-\log L(\theta_0))\\
				\nonumber &+\log L(\theta_0')-\log L(\theta_0))\\
				\nonumber &\leq \sup_{ \{\beta -\beta_0\in [\epsilon_n,\epsilon], B-B_0 \in [0,\epsilon]\}}
				(\log L(\theta)-\log L(\theta_0))+\sup_{ \{\beta\geq \beta_0, B >B_0+\epsilon\}}(\log L(\theta)-\log L(\theta_0))\\
				\nonumber&+\log L(\theta_0')-\log L(\theta_0))\\
				\nonumber&\leq \sup_{ \{\beta -\beta_0\in [\epsilon_n,\epsilon], B-B_0 \in [0,\epsilon]\}}(\log L(\theta)-\log L(\theta_0))+\log L(\theta_0'')-\log L(\theta_0)\\
				\nonumber&+\log L(\theta_0')-\log L(\theta_0)\\
				\label{e:upp-lnl0}&\leq  \sup_{ \{\beta -\beta_0\in [\epsilon_n,\epsilon], B-B_0 \in [0,\epsilon]\}}3(\log L(\theta)-\log L(\theta_0))\leq -C_0n\epsilon_n^2
			\end{align}
			where the second inequality follows from (a) and fifth inequality follows from (b) above.  Finally for the last inequality, consider Taylor expansion for $\log L(\theta)$ upto the second order
			\begin{align*}
				\log L(\theta) - \log L(\theta_0) &=W_n(\theta_0\mid \boldx)^\top (\theta-\theta_0)-\frac{1}{2}(\theta - \theta_0)^\top H_n(\tilde{\theta}\mid \boldx) (\theta - \theta_0)
			\end{align*}
			where $\tilde{\theta}=\theta_0+c(\theta-\theta_0)$, $0<c<1$ and $W_n$ and $H_n$ are as defined in \eqref{e:first-order} and \eqref{e:second-order} respectively.
			
			\noindent By Cauchy Schwarz inequality,
			\begin{align*}
				|W_n(\theta_0\mid \boldx)^\top (\theta-\theta_0))|&\leq  (\left(W_{1n}(\theta_0\mid \boldx)^2 + W_{2n}(\theta_0\mid \boldx)^2\right) \lVert\theta - \theta_0\rVert_2^2 +1)\\
				&\leq n^{2/3}  \lVert\theta - \theta_0\rVert_2^2+1
			\end{align*}
			for every $\boldx \in \mathcal{A}_n$. Further \begin{align*}
				\log L(\theta) - \log L(\theta_0)  &\leq  n^{2/3}\lVert\theta - \theta_0\rVert_2^2+1 -\frac{1}{2}(\theta - \theta_0)^\top H_n(\tilde{\theta}\mid \boldx) (\theta - \theta_0)\\
				&\leq n^{2/3}\lVert\theta - \theta_0\rVert_2^2+1 - \frac{\text{sech}^4(\tilde{\beta} \gamma + |\tilde{B}|)}{1+\gamma^2}nT_n(\boldx) \lVert\theta - \theta_0\rVert_2^2 \\
				&\leq \left(n^{2/3}+\frac{1}{||\theta-\theta_0||_2^2}- \frac{\text{sech}^4(\tilde{\beta} \gamma + |\tilde{B}|)}{1+\gamma^2}\frac{n}{M}  \right) \lVert\theta - \theta_0\rVert_2^2 
			\end{align*}
			where the second inequality is a consequence of the lower bound \eqref{eigen_lowerb} and the third inequality holds since $\boldx \in \mathcal{A}_n$. Taking $\sup$ over the set $\{\beta -\beta_0\in [\epsilon_n,\epsilon], B-B_0 \in [0,\epsilon]\}$ on both sides, 
			\begin{align}
				\label{e:ll-up}
				\nonumber   &\sup_{ \{\beta -\beta_0\in [\epsilon_n,\epsilon], B-B_0 \in [0,\epsilon]\}}(\log L(\theta)-\log L(\theta_0))\\
				\nonumber   &\leq \sup_{ \{\beta -\beta_0\in [\epsilon_n,\epsilon], B-B_0 \in [0,\epsilon]\}} \left(n^{2/3}+\frac{1}{\epsilon_n^2}- \frac{\text{sech}^4((\beta_0+\epsilon) \gamma + (B_0+\epsilon))}{1+\gamma^2}\frac{n}{M}  \right) \lVert\theta - \theta_0\rVert_2^2 \\
				&\leq -C_0n\epsilon_n^2
			\end{align} 
			for some $C_0>0$ as $n \to \infty$ since $n^{2/3}$ and $1/\epsilon_n^2=o(n)$.This completes the proof.
		\end{proof}
		
		\subsection{Technical details of Lemma \ref{lem:E2}}
		
		\begin{lemma}
			\label{dif_true_true}
			Let $L_0$ and $L(\theta_0)$ represent the true likelihood \eqref{likel} and the pseudo-likelihood \eqref{pseudo} with the true parameters, respectively. Then,
			\begin{align*}
				\frac{1}{n\epsilon_n^2}    \E_0^{(n)}(\log L_0-\log L(\theta_0))\to 0,\:\: n \to \infty
			\end{align*}
		\end{lemma}
		
		\begin{proof}
			\vspace{-6mm}
			$$L_0=\frac{e^{f_{\theta_0}(\boldx)}}{\sum_{\boldx \in \{-1,1\}^n} e^{f_{\theta_0}(\boldx)}}=\frac{e^{f_{\theta_0}(\boldx)}}{Z_n(\theta_0)}$$
			where $f_{\theta_0}(\boldx)=(\beta_0/2)\boldx^\top A_n \boldx+B_0 \boldx^\top1$.
			\noindent Define $b(\boldx;\theta)=(b_1(\boldx; \theta), \cdots, b_n(\boldx; \theta))$ where
			$$b_i(\boldx; \theta)=E(X_i|X_j, j\neq i)=\tanh(\beta m_i(\boldx)+B)$$
			\noindent  Then $L(\theta)=e^{g(\boldx, b(\boldx; \theta))}$
			where the function $g$ for $\boldv,\boldw \in [-1,1]^n$ is defined as
			$$g(\boldv,\boldw)=\sum_{i=1}^n \frac{1+v_i}{2} \log \frac{1+w_i}{2}+\frac{1-v_i}{2}\log \frac{1-w_i}{2}$$
			Also, define $I(\boldv)=g(\boldv,\boldv)$. 
			Now, observe that
			\begin{align}
				\label{e:split-0}
				\nonumber \E_0^{(n)}(\log L_0-\log L(\theta_0))&=\E_0^{(n)}(f_{\theta_0}(\boldx)-g(\boldx,b(\boldx;\theta_0))-\log Z_n(\theta_0))\\
				\nonumber    &\hspace{-20mm}=\E_0^{(n)}(f_{\theta_0}(\boldx)-f_{\theta_0}(b(\boldx;\theta_0))+\E_0^{(n)}(f_{\theta_0}(b(\boldx;\theta_0))-I(b(\boldx;\theta_0)))\\
				\nonumber    &\hspace{-20mm}+\E_0^{(n)}(I(b(\boldx;\theta_0))-g(x,b(\boldx;\theta_0)))-\log Z_n(\theta_0)\\
				\nonumber    & \hspace{-20mm}\leq (\E_0^{(n)}(f_{\theta_0}(\boldx)-f_{\theta_0}(b(\boldx;\theta_0)))^2)^{1/2}+(\E_0^{(n)}(I(b(\boldx;\theta_0))-g(\boldx,b(\boldx;\theta_0)))^2)^{1/2}\\
				&\hspace{-20mm}+\E_0^{(n)}(f_{\theta_0}(b(\boldx;\theta_0))-I(b(\boldx;\theta_0)))-\log Z_n(\theta_0),
			\end{align}
			where the last step is due to Hölder's inequality.
			
			\noindent Under Assumption \ref{assumption_mf}, mimicking the proof of Lemmas 3.2 and 3.3 in \cite{Basak:2017} with $n$ replaced by $n\epsilon_n^2$, we get
			\begin{align}
				\label{e:split-1}
				(\E_0^{(n)}(f_{\theta_0}(\boldx)-f_{\theta_0}(b(\boldx;\theta_0)))^2)^{1/2}=o(n\epsilon_n^2)
			\end{align}
			\begin{align}
				\label{e:split-2}
				(\E_0^{(n)}(I(b(\boldx;\theta_0))-g(\boldx,b(\boldx;\theta_0)))^2)^{1/2}=o(n\epsilon_n^2)
			\end{align}
			Also for  $r_n=\sup_{\boldv\in {[-1,1]}^n} (f_{\theta_0}(\boldv)-I(\boldv))$, we have
			$$\E_0^{(n)}(f_{\theta_0}(b(\boldx,\theta_0))-I(b(\boldx,\theta_0)))\leq r_n$$
			
			\noindent By Theorem 1.6 in \cite{chatterjee2016nonlinear} with the fact  $\partial^2 f_{\theta_0}/\partial x_i^2  = 0$, $i=1, \cdots, n$, we have $ -\log Z_n(\theta_0)\leq -r_n$. Therefore,
			\begin{align}
				\label{e:split-3}
				\E_0^{(n)}(f_{\theta_0}(b(\boldx;\theta_0))-I(b(\boldx;\theta_0))-\log Z_n(\theta_0))\leq 0.
			\end{align}
			Using \eqref{e:split-1}, \eqref{e:split-2} and \eqref{e:split-3} in \eqref{e:split-0} completes the proof.
		\end{proof}
		
		\begin{lemma}
			\label{constant-order}
			Note that $L(\theta)$ is not a valid density function. So, we consider  $\tilde{L}(\theta)=L(\theta)/J_n(\theta)$
			where $J_{n}(\theta) = \sum_{\boldx \in \{-1,1\}^n}L(\theta)$ such that $\sum_{\boldx \in \{-1,1\}^n}\tilde{L}(\theta) = 1$. Then for every  $\theta$, 
			$$J_n(\theta)\leq \beta \epsilon_n\sqrt{n(1+\gamma^2)/2}+o(n\epsilon_n^2)(\log 3\sqrt{2}-\log \epsilon_n).$$
		\end{lemma}
		
		\begin{proof}
			Let $N_n(\epsilon_n):=\{i \in [n]:|\lambda_i(A_n)|>\epsilon_n/\sqrt{2}\}$ and with the mean field condition in the Assumption \ref{assumption_mf}, it is easy to note that
			\begin{align}
				\frac{\left|N_{n}(\epsilon_n)\right|}{n} \leq \frac{2}{n \epsilon_n^{2}} \sum_{i \in[n]} \lambda_{i}\left(A_{n}\right)^{2}=\frac{2}{n \epsilon_n^{2}} \sum_{i, j=1}^{n} A_{n}(i, j)^{2} \rightarrow 0, \:\: n \rightarrow \infty
			\end{align}
			Set $k_n=|N_n(\epsilon_n)|$ and let $D_{n,0}(\epsilon_n)$ be a $\epsilon_n\sqrt{n/2}$ net of the set $\{\boldf \in \mathbb{R}^{k_n}: \sum f_i^2\leq n\}$ of size at most $(3\sqrt{2}/\epsilon_n)^{k_n}$. The existence of such a net is standard (see for example Lemma 2.6 in \cite{milman1986asymptotic}).
			
			\noindent Let $\{\boldp_1, \cdots, \boldp_n\}$ be the eigen vectors of $A_n$. Then setting 
			$$D_{n,1}(\epsilon_n):=\{\sum_{i\in N_n(\epsilon_n)} c_i\lambda_i(A_n)\boldp_i, \boldc \in D_{n,0}(\epsilon_n)\}$$
			We claim $D_{n,1}(\epsilon_n)$ is $\epsilon_n\sqrt{n(1+\gamma^2)/2}$ of the set $\{A_n \boldx: \boldx \in \{-1,1\}^n\}$. 
			Indeed any $\boldx \in \{-1,1\}^n$ can be written as $\sum_{i=1}^n f_i \boldp_i$ where $\sum_{i=1}^n f_i^2=\sum x_i^2 =n$. In particular, it means $\sum_{i \in N_n(\epsilon_n)} f_i\leq n$, which  implies there exists a $\boldc \in D_{n,0}(\epsilon_n)$ such that $||\boldc-\boldf||\leq \epsilon_n\sqrt{n/2}$. 
			
			\noindent Let $\sum_{i \in N_n(\epsilon_n)}c_i\lambda_i(A_n)\boldp_i \in D_{n,1}(\epsilon_n)$, then
			\begin{align*}
				||A_n \boldx-\sum_{i \in N_n(\epsilon_n)}c_i \lambda_i(A_n)\boldp_i||_2^2&=\sum_{i \in N_n(\epsilon_n)} (c_i-f_i)^2 \lambda_i(A_n)^2+\sum_{i \notin N_n(\epsilon_n)} \lambda_i(A_n)^2 f_i^2\\
				&\leq \frac{\gamma^2 n\epsilon_n^2}{2}+\frac{n\epsilon_n^2}{2}
			\end{align*}
			where the last inequality is a consequence of $\max_{i \in [n]}|\lambda_i(A_n)|\leq \max_{i \in [n]} \sum_{j=1}^n|A_n(i,j)|\leq \gamma$ and the definition of the set $N_n(\epsilon_n)$.
			
			In particular for any $\boldx \in \{-1,1\}^n$, there exists at least one $\boldp \in D_{n,1}(\epsilon_n)$ such that $||\boldp-m(\boldx)||\leq \epsilon_n\sqrt{n(1+\gamma^2)/2}$. For any $\boldp \in D_{n,1}(\epsilon_n)$, let
			$$\mathcal{P}(\boldp):=\{\boldx\in \{-1,1\}^n:||\boldp-m(\boldx)||\leq \epsilon_n\sqrt{n(1+\gamma^2)/2} \}$$
			Therefore,
			$$\sum_{\boldx \in \{-1,1\}^n }e^{g(\boldx,b(\boldx;\theta))}=\sum_{\boldp \in D_{n,1}(\epsilon_n)} \sum_{\boldx \in \mathcal{P}(\boldp)}e^{g(\boldx,b(\boldx;\theta))}$$
			Setting $\boldu(\boldp):=\tanh(\beta \boldp+B)$ if $||\boldp-m(\boldx)||\leq \epsilon_n\sqrt{n(1+\gamma^2)/2}$, then we have
			$$|g(\boldx,b(\boldx;\theta))-g(\boldx,\boldu(\boldp))|\leq 2\beta\sum_{i=1}^n |m_i(\boldx)-p_i|\leq 2\beta\epsilon_n\sqrt{n(1+\gamma^2)/2}$$
			Finally,  \begin{align*}\sum_{\boldx \in \{-1,1\}^n }e^{g(\boldx,b(\boldx;\theta))}&\leq e^{\beta \sqrt{n}\epsilon_n}\sum_{\boldp \in D_{n,1}(\epsilon_n)} \sum_{\boldx \in \mathcal{P}(\boldp)}e^{g(\boldx,u(\boldp))}\\
				&\leq e^{\beta \sqrt{n}\epsilon_n}\sum_{\boldp \in D_{n,1}(\epsilon_n)} \sum_{\boldx \in \{-1,1\}^n}e^{g(\boldx,u(\boldp))}=e^{\beta \epsilon_n\sqrt{n(1+\gamma^2)/2}} |D_{n,1}(\epsilon_n)|
			\end{align*}
			where the last equality follows since $\sum_{\boldx \in \{-1,1\}^n} e^{g(x,u)}=1$ for any $u \in [-1,1]^n$. Therefore,
			$$\log J_n(\theta) \leq \beta \epsilon_n\sqrt{n(1+\gamma^2)/2}+\log |D_{n,1}(\epsilon_n)|$$
			Since $|D_{n,1}(\epsilon_n)|=|D_{n,0}(\epsilon_n)|$, therefore
			$$\log|D_{n,1}(\epsilon_n)| \leq |N_n(\epsilon_n)| (\log 3\sqrt{2}-\log \epsilon_n)$$
			The proof follows since $|N_n(\epsilon_n)|=o(n\epsilon_n^2)$.
		\end{proof}
		
		\begin{lemma}
			\label{subset}
			Define $\mathcal{V}_{\epsilon_n} := \{\theta: |\beta - \beta_0| < \epsilon_n, |B - B_0| < \epsilon_n \}$. Then,
			\begin{align*}
				\mathcal{V}_{\epsilon_n} \subseteq \mathcal{K}_{\epsilon_n}, \:\:\text{for $n$ sufficiently large}
			\end{align*}
			where $\mathcal{K}_{\epsilon_n}:=\{\theta: \E_0^{(n)}(\log (L(\theta_0)/L(\theta)))<3n\epsilon_n^2\}$.
		\end{lemma}
		\begin{proof} For any $\theta \in \mathcal{V}_{\epsilon_n}$, using the decomposition in \eqref{e:Ltheta-L0}, we get
			\begin{align*}
				\E_0^{(n)}\left( \log L(\theta_0) - \log L(\theta) \right) &= \E_0^{(n)} \left( - \circled{1} - \circled{2} +  \circled{3} - \circled{4} \right)\leq 3n\epsilon_n^2
			\end{align*}
			where the last inequality is justified next.
			
			\noindent For some $M>0$ using Lemma \ref{qn:rn:bounded}, we get
			\begin{align*}
				-\E_0^{(n)}(\circled{1})&=(\beta_0-\beta)\E_0^{(n)}(W_{1n}(\theta_0|\boldx)\leq  \sqrt{n}|\beta_0 - \beta| \left(\frac{1}{n}\E_0^{(n)}\left(W_{1n}(\theta_0|\boldx)\right)^2\right)^{1/2} \\
				&\leq M\sqrt{n}\epsilon_n
			\end{align*}
			\begin{align*}
				-\E_0^{(n)}(\circled{2})&=(B_0-B)\E_0^{(n)}(W_{2n}(\theta_0|\boldx)\leq  \sqrt{n}|B_0 - B| \left(\frac{1}{n}\E_0^{(n)}\left(W_{2n}(\theta_0|\boldx)\right)^2\right)^{1/2}\\
				&\leq M\sqrt{n}\epsilon_n
			\end{align*}
			By relation \eqref{e:hess-upp}, we get
			$$\E_0^{(n)}(\circled{3})\leq  n||\theta-\theta_0||_2^2
			\left((1+\gamma^2) - \frac{\text{sech}^4(\beta_0 \gamma + |B_0|)}{1+\gamma^2}\E_0^{(n)}\left(T_n(\boldx)\right) \right) \leq 2(1+\gamma^2)n\epsilon_n^2  $$
			By relation \eqref{e:fourth-order}, we get
			$$-\E_0^{(n)}(\circled{4})\leq  \frac{0.38n}{3(1 + \gamma^2)} \E_0^{(n)}\left(T_n(\boldx)\right)||\theta-\theta_0||_2^2(|\beta - \beta_0|\gamma + |B - B_0|)\leq \frac{0.38\gamma^2(1+\gamma)}{3(1 + \gamma^2)}n\epsilon_n^3 $$
		\end{proof}
		
		\begin{lemma}
			\label{prior}
			With prior distribution $p(\theta)$ as in \eqref{e:prior}, we have
			$$\int_{\mathcal{V}_{\epsilon_n}} p(\theta) d\theta \geq C\epsilon_n^2, \hspace{5mm}\text{for some $C>0$}$$
		\end{lemma}
		
		\begin{proof} 
			By mean value theorem with $\beta^\star \in [\beta_0 - \epsilon_n, \beta_0 + \epsilon_n]$ and $B^\star \in [B_0 - \epsilon_n, B_0 + \epsilon_n]$,
			\begin{align*}
				\int_{V_{\epsilon_n}} p(\theta) d\theta &= \int_{\beta_0 - \epsilon}^{\beta_0 + \epsilon}  \frac{1}{\beta \sqrt{2\pi}}e^{-\frac{(\log\beta)^2}{2}}d\beta  \int_{B_0 - \epsilon}^{B_0 + \epsilon}  \frac{1}{\sqrt{2\pi}}e^{-\frac{B^2}{2}}dB \\
				&= \frac{2\epsilon_n}{\beta^\star \sqrt{2\pi}}e^{-\frac{(\log \beta^\star)^2}{2}} \frac{2\epsilon_n}{\sqrt{2\pi}}e^{-\frac{(B^\star)^2}{2}} \\
				&= \exp\left(-(\log\pi - \log 2 -2\log\epsilon_n) - \frac{1}{2}\left(2\log \beta^\star + (\log \beta^\star)^2 + (B^\star)^2\right)\right) \\
				&\geq \exp\left(-( \log\pi - \log 2 -2\log\epsilon_n) - \frac{1}{2}\left(2u_1 + \tilde{u}_1 + u_2\right)\right) \\
				&\geq C e^{2\log \epsilon_n} = C\epsilon_n^2
			\end{align*} 
			where the above result follow since $\epsilon_n \to 0$ implies  $u_1 \leq  \max(\log(\beta_0+1), \log(\beta_0 +1))$, $\tilde{u}_1 \leq \max((\log(\beta_0-1))^2, (\log(\beta_0 + 1))^2 )$, and $u_2 = \max((B_0-1)^2, (B_0 + 1)^2 )$. 
		\end{proof}
		\begin{proof} [Proof of Lemma \ref{lem:E2}]
			Let $L^*=\int  L(\theta) p(\theta)d\theta$, $J_n^* = \sum_{\boldx \in \{-1,1\}^n} L^*$.
			Then $J_n^*=\sum_{\boldx \in \{-1,1\}^n} \int L(\theta)p(\theta)d\theta$. Since $L(\theta)p(\theta)>0$, Tonelli's theorem allows for interchange of the order of summation and integral. Using Lemma \ref{constant-order} and $-\log \epsilon_n=O(\log n)$,  we get
			\begin{align}
				\label{e:Jn-order} J_n^*&=\int \sum_{\boldx \in \{-1,1\}^n} L(\theta) p(\theta)d\theta=\int J_n(\theta)p(\theta)d\theta \nonumber \\
				&=\epsilon_n \sqrt{n(1+\gamma^2)/2}E_P(\beta)+o(n\epsilon_n^2)(\log 3\sqrt{2}-\log \epsilon_n) \nonumber\\
				&=\epsilon_n \sqrt{ne(1+\gamma^2)/2}+o(n\epsilon_n^2)(\log 3\sqrt{2}-\log \epsilon_n)=o(n\epsilon_n^2 \log n)
			\end{align}
			Also, by Lemma \ref{constant-order} and $-\log \epsilon_n=O(\log n)$,
			\begin{equation}\label{e:Jn-0-order}
				\log J_n(\theta_0)=\beta_0 \epsilon_n \sqrt{n(1+\gamma^2)/2}+o(n\epsilon_n^2)(\log 3\sqrt{2}-\log \epsilon_n)=o(n\epsilon_n^2 \log n)
			\end{equation}
			\begin{align}
				\label{E2_Lee}
				&P_0^n\left(\left|\log\int  ( L(\theta)  /L(\theta_0)) p(\theta)d\theta\right|>C n\epsilon_n^2 \log n\right)\\
				&\leq \frac{1}{C n\epsilon_n^2 \log n}\E_0^{(n)}\left(\left|\log\int  ( L(\theta)  /L(\theta_0)) p(\theta)d\theta\right|\right) \nonumber \\	&=\frac{1}{ C n\epsilon_n^2 \log n} \E_0^{(n)} \left( \left|\log \left(L^*/L(\theta_0)\right) \right|\right) \nonumber\\
				&\leq \frac{1}{C n\epsilon_n^2 \log n} \left({\rm KL}(L_0, \tilde{L}^*)+{\rm KL}(L_0, \tilde{L}(\theta_0))+\left|\log \frac{J_n^*}{J_{n}(\theta_0)} \right|+ \frac{4}{e} \right) \nonumber \\
				\nonumber &\leq \frac{2}{C n\epsilon_n^2\log n} \Big( \E_0^{(n)}\left(\log L_0-\log L(\theta_0) \right) + \E_0^{(n)}(\log L(\theta_0)-\log L^*)\\
				&+2(\log J_n^* + \log J_{n}(\theta_0)) + \frac{4}{e} \Big)
			\end{align}
			where the second last step follows from Lemma \ref{Lee}.
			
			Then, using the set $\mathcal{K}_{\epsilon_n}$ in Lemma \ref{subset}, we get
			\begin{align}
				\label{e:b-bound-2}
				\nonumber	&\E_0^{(n)}(\log L(\theta_0)-\log L^*))\\
				\nonumber	&=\E_0^{(n)}(\log L(\theta_0)-\log \int  L(\theta) p(\theta)d\theta)\\
				\nonumber	&\leq \E_0^{(n)}\left(\log(L(\theta_0)/\int_{\mathcal{K}_{\epsilon_n^2}}  L(\theta) p(\theta)d\theta)\right)\\
				\nonumber&\leq \E_0^{(n)}\left(\log L(\theta_0)-\log\left(\frac{p(\mathcal{K}_{ \epsilon_n^2})}{p(\mathcal{K}_{ \epsilon_n^2)}} \int_{\mathcal{K}_{ \epsilon_n^2}}  L(\theta) p(\theta)d\theta) )\right)\right)\\
				\nonumber&\leq \E_0^{(n)}(\log L(\theta_0))-\log (p(\mathcal{K}_{ \epsilon_n^2})+\E_0^{(n)} \left(E_{p|\mathcal{K}_{ \epsilon_n^2} }(-\log L(\theta))\right)\\
				\nonumber&\leq -\log (p(\mathcal{K}_{ \epsilon_n^2}))+\E_0^{(n)} \left(\log L(\theta_0)- \int_{\mathcal{K}_{ \epsilon_n^2}}  \log L(\theta) p|\mathcal{K}_{ \epsilon_n^2}(\theta) d\theta\right) \:\: \text{Jensen's Inequality}\\
				\nonumber&=-\log (p(\mathcal{K}_{ \epsilon_n^2}))+\int_{\mathcal{K}_{ \epsilon_n^2}} \E_0^{(n)} (\log (L(\theta_0)-L(\theta)) p|\mathcal{K}_{ \epsilon_n^2}(\theta) d\theta\\
				&\leq -2\log (C' \epsilon_n^2)+3n \epsilon_n^2=o(n\epsilon_n^2 \log n)
			\end{align}
			where  the last line follows from Lemma \ref{subset} and Lemma \ref{prior}.
			The final order is because $-\log \epsilon_n = O(\log n)$ and $n\epsilon_n^2 \to \infty$ and $\log n \to \infty$.
			
			\noindent The proof follows by using relations \eqref{e:b-bound-2}, \eqref{e:Jn-0-order} and \eqref{e:Jn-order} in \eqref{E2_Lee}.
		\end{proof}
		
		\subsection{Technical details of Lemma \ref{lem:E3}}
		
		\begin{lemma}
			\label{q_exists}
			Let $q(\theta) \in \mathcal{Q}^{MF}$ with $\mu_1 = \log\beta_0$, $\mu_2 = B_0$, and $\sigma_1^2 = \sigma_2^2 = 1/n$, then 
			\begin{align*}
				\frac{1}{n\epsilon_n^2}\int \E_0^{(n)}(\log L(\theta_0)-\log L(\theta)) q(\theta)  d\theta \lesssim 0,  \:\: n \to \infty
			\end{align*}
		\end{lemma}
		
		\begin{proof} 
			Using the Taylor expansion of $\log L(\theta)$ around $\theta=\theta_0$, we get
			\begin{align}
				\label{e:Ltheta-L0}
				\nonumber  \log L(\theta_0) - \log L(\theta) &=  \log L(\theta_0) -\log L(\theta_0) - \underbrace {(\beta-\beta_0)W_{1n}(\theta_0|\boldx)}_{\circled{1}}- \underbrace{W_{2n}(\theta_0|\boldx)(B - B_0)}_{\circled{2}}\\
				&+  \underbrace{\frac{1}{2}(\theta - \theta_0)^{\top} H_n(\theta_0|\boldx) (\theta - \theta_0)}_{\circled{3}} - \underbrace{R_n(\tilde{\theta}, \theta-\theta_0|\boldx)}_{\circled{4}} 
			\end{align}
			$W_{1n}$, $W_{2n}$ is as in \eqref{e:first-order}, $H_{n}$ is  as in \eqref{e:second-order} and $R_n(\tilde{\theta},\theta-\theta_0|\boldx)$ is defined in \eqref{e:third-order}.
			Therefore,
			\begin{align*}
				&\int  \E_0^{(n)}\left(\log L(\theta_0)-\log L(\theta)\right) q(\theta) d\theta  \\
				&=  - \int  \E_0^{(n)}\left(\circled{1}\right)q(\theta) d\theta - \int  \E_0^{(n)}\left(\circled{2}\right)q(\theta) d\theta \\
				&+ \int \E_0^{(n)}\left(\circled{3}\right)q(\theta) d\theta - \int \E_0^{(n)}\left(\circled{4}\right)q(\theta) d\theta,
			\end{align*}
			
			\begin{align*}
				- \frac{1}{n\epsilon_n^2}    \int \E_0^{(n)}\left(\circled{1}\right)q(\theta) d\theta &= \frac{1}{n\epsilon_n^2}\int (\beta_0 - \beta) \E_0^{(n)}\left(W_{1n}(\theta_0|\boldx)\right)q(\theta) d\theta \\
				&\leq \frac{1}{n\epsilon_n^2}\int |\beta_0 - \beta| \sqrt{n}\left(\frac{1}{n}\E_0^{(n)}\left(W_{1n}(\theta_0|\boldx)\right)^2\right)^{1/2}q(\beta) d\beta,\\
				&\leq \frac{M \sqrt{n}}{n\epsilon_n^2} \int |\beta -\beta_0| q_\beta(\beta) d\beta \\
				&\leq \frac{M \sqrt{n}}{n\epsilon_n^2}  (\int (\beta -\beta_0)^2 q_\beta(\beta) d\beta)^{1/2}\\
				&= \frac{M\sqrt{n}}{n\epsilon_n^2}  ( e^{2\log\beta_0}(e^{2/n} - 2e^{1/2n} + 1))^{1/2}\sim  \frac{M e^{\log \beta_0}  }{n\epsilon_n^2} \to 0
			\end{align*}
			where the second inequality above above line holds by H\"{o}lder inequality and third inequality holds by Lemma \ref{qn:rn:bounded}, for some constant $M$. Finally the last convergence to 0 is  $n\epsilon_n^2 \to \infty$. Similarly,
			\begin{align*}
				-\frac{1}{n\epsilon_n^2}    \int \E_0^{(n)}\left(\circled{2}\right)q(\theta) d\theta &= \frac{1}{n\epsilon_n^2}  \int (B_0 - B) \E_0^{(n)}\left(W_{2n}(\theta_0|\boldx)\right)q(\theta) d\theta \\
				&\leq \frac{\sqrt{n}}{n\epsilon_n^2}  \left(\frac{1}{n}\E_0^{(n)}\left(W_{2n}(\theta_0|\boldx)\right)^2\right)^{1/2} \int |B -B_0| q_B(B) dB \\
				& \frac{M\sqrt{n}}{n\epsilon_n^2}  (\int (B-B_0)^2 q(B)dB)^{1/2} \sim \frac{M}{n\epsilon_n^2} \to 0
			\end{align*}
			Using the upper bound \eqref{e:hess-upp} and $n\epsilon_n^2 \to \infty$, we get
			\begin{align*}
				&\frac{1}{n\epsilon_n^2}   \int \E_0^{(n)}\left(\circled{3}\right)q(\theta) d\theta \\
				&\leq \frac{1}{2\epsilon_n^2}\left((1+\gamma^2) - \frac{\text{sech}^4(\beta_0 \gamma + |B_0|)}{1+\gamma^2}\E_0^{(n)}\left(T_n(\boldx)\right) \right) \int \left\{(\beta - \beta_0)^2 + (B - B_0)^2\right\} q(\theta) d\theta\\
				&\leq  \frac{(1+\gamma^2)}{2\epsilon_n^2} \left( e^{2\log\beta_0}\left(e^{2/n} - 2e^{1/2n} + 1\right) + 1/n\right)\\
				&\sim \frac{(e^{2\log \beta_0}+1)(1+\gamma^2)}{2n\epsilon_n^2} \to 0
			\end{align*}
			where the second inequality holds  since $T_n(\boldx) = (1/n)\sum_{i=1}^n \left( m_i(\boldx) - (1/n) m_i(\boldx)\right)^2 \geq 0$.
			
			\noindent For the remainder term, using relations \eqref{Cn_lower_b} and \eqref{Dn_lower_b} in relation  \eqref{e:third-order}, we get
			\begin{align}
				\label{e:fourth-order}
				-\E_0^{(n)}&\left(\circled{4}\right)\leq \frac{0.38n}{3(1 + \gamma^2)} \E_0^{(n)}\left(T_n(\boldx)\right)\left\{(\beta - \beta_0)^2 + (B - B_0)^2\right\}\left\{ (\beta - \beta_0)\gamma + (B - B_0)\right\}
			\end{align}
			Further,
			\begin{align}
				& n \int \left\{(\beta - \beta_0)^2 + (B - B_0)^2\right\}\left\{ (\beta - \beta_0)\gamma + (B - B_0)\right\}q(\theta) d\theta \nonumber \\
				&= \underbrace{n\int \left\{(e^{\log\beta} - e^{\log\beta_0})^2 + (B - B_0)^2\right\}\left\{ (e^{\log\beta} - e^{\log\beta_0})\gamma + (B - B_0)\right\}q(\theta) d\theta}_{\circled{5}} \nonumber     \end{align}
			\begin{align}
				\circled{5}
				&= n\int \left( e^{3\log\beta} - 3e^{2\log\beta + \log\beta_0} + 3e^{\log\beta + 2\log\beta_0} - e^{3\log\beta_0}\right)q(\beta) d\beta \label{term1}\\
				&  + n\int \left( B^3 - 3B_0B^2 + 3B_0^2B - B_0^3\right)q(B)  dB  \label{term2}\\
				& + n\int \left( e^{2\log\beta} - 2e^{\log\beta + \log\beta_0} + e^{2\log\beta_0} \right)\left(B - B_0 \right)q(\beta)q(B) d\beta dB \label{term3}\\
				&  + n\gamma \int \left(B - B_0\right)^2\left(e^{\log\beta} - e^{\log\beta_0}\right) q(\beta)q(B) d\beta dB \label{term4}  
			\end{align}
			\begin{align*}
				\eqref{term1} &= ne^{3\log\beta_0}\left( e^{9/2n} - 3e^{2/n} + 3e^{1/2n} - e^{0}\right) \sim  0 \\
				\eqref{term2} &= n\left( B_0^3 + 3B_0/n -3B_0(B_0^2 + 1/n) +3B_0^3 - B_0^3\right) = 0,
				\\
				\\
				\eqref{term3} &= n\left(B_0e^{2\log\beta_0}\left( e^{2/n} -2 e^{1/2n} + e^0 \right) - B_0e^{2\log\beta_0}\left( e^{2/n} -2 e^{1/2n} + e^0 \right)\right) = 0, \\
				\\
				\eqref{term4} &= ne^{\log\beta_0}\left(\frac{\gamma}{n}e^{1/2n} - \frac{\gamma}{n}e^0 \right) = \gamma e^{\log\beta_0}(e^{1/2n} - e^0 ) \sim 0
			\end{align*}
			\begin{align*}
				-\frac{1}{n\epsilon_n^2} \int \E_0^{(n)}\left(\circled{4}\right)q(\theta) d\theta \lesssim  0
			\end{align*}
			since $T_n(\boldx)\leq  \gamma^2$ (since by Assumption \ref{assumption_rs}, $m_i(\boldx)\leq  \gamma$) and $n\epsilon_n^2 \to \infty$.
		\end{proof}
		
		\noindent 
		\begin{proof} [Poof of Lemma \ref{lem:E3}] 
			
			\noindent With the $q$ as in Lemma \ref{q_exists}, using Markov's inequality,  
			\begin{align}
				\label{E3_detail}
				&P_0^n\left(\int  q(\theta) \log (L(\theta_0)/ L(\theta))d\theta> C  n\epsilon_n^2\log n \right)\\
				\nonumber & \leq \frac{1}{C n\epsilon_n^2\log n}\E_0^{(n)}\left|\int  q(\theta) \log (L(\theta_0)/ L(\theta))d\theta\right|\\
				\nonumber &\leq \frac{1}{C n\epsilon_n^2\log n}\E_0^{(n)}\left(\int  q(\theta) \left|\log  (L(\theta_0)/ L(\theta))\right|d\theta\right)\\
				\nonumber &\leq \frac{1}{C n\epsilon_n^2\log n} \int  q(\theta)\E_0^{(n)}\left( \left|\log (L(\theta_0)/ L(\theta))\right| \right) d\theta \qquad \text{Fubini's theorem} \\
				\nonumber &=\frac{1}{C n\epsilon_n^2\log n} \int  q(\theta)\E_0^{(n)}\left( \left| \log\left( \frac{L_0}{L(\theta)}\frac{L(\theta_0)}{L_0}\right) \right| \right) d\theta\\
				\nonumber &=\frac{1}{C n\epsilon_n^2\log n} \int  q(\theta)\E_0^{(n)}\left( \left|\log \left(\frac{L_0}{\tilde{L}(\theta)}\frac{\tilde{L}(\theta_0)}{L_0} \frac{J_n(\theta_0)}{J_n(\theta)} \right)\right| \right) d\theta\\
				\nonumber &=\frac{1}{C n\epsilon_n^2\log n} \int  q(\theta)\E_0^{(n)}\left( \left|\log \left(\frac{L_0}{\tilde{L}(\theta)}\right) + \log\left(\frac{\tilde{L}(\theta_0)}{L_0}\right) + \log\left( \frac{J_n(\theta_0)}{J_n(\theta)} \right)\right| \right) d\theta \\
				&\nonumber \leq\frac{1}{C n\epsilon_n^2\log n} \int  q(\theta)\E_0^{(n)}\left( \left|\log \left(\frac{L_0}{\tilde{L}(\theta)}\right)\right| + \left|\log\left(\frac{\tilde{L}(\theta_0)}{L_0}\right)\right| + \left|\log\left( \frac{J_n(\theta_0)}{J_n(\theta)} \right)\right| \right) d\theta\\
				&\leq\frac{1}{C n\epsilon_n^2\log n} \int  q(\theta)\left( \text{KL}\left( L_0, \tilde{L}(\theta) \right) + \text{KL} \left( L_0, \tilde{L}(\theta_0) \right)  + \left|\log\left( \frac{J_n(\theta_0)}{J_n(\theta)} \right)\right| +\frac{4}{e} \right) d\theta 
			\end{align}
			Therefore,
			\begin{align}
				\nonumber
				&P_0^n\left(\int  q(\theta) \log (L(\theta_0)/ L(\theta))d\theta> C  n\epsilon_n^2\log n \right)\\
				&= \frac{1}{C n\epsilon_n^2\log n} \left( 2\E_0^{(n)}(\log L_0 - \log L(\theta_0)) + \int q(\theta) \E_0^{(n)}(\log L(\theta_0)-\log L(\theta)) d\theta \right) \nonumber \\
				& + \frac{1}{C n\epsilon_n^2\log n}\left(2\log J_n(\theta_0) + 2\int q(\theta) \log J_n(\theta) d\theta + \frac{4}{e} \right) \to 0
			\end{align} 
			where the inequality in second last step  is due to Lemma \ref{Lee}. The last convergence to 0 is explained next.  By Lemma \ref{dif_true_true} and Lemma \ref{q_exists} respectively, we get \begin{align*}
				\E_0^{(n)}(\log L_0-\log L(\theta_0)) &=o(n\epsilon_n^2)\\
				\int q(\theta) \E_0^{(n)}(\log L(\theta_0)-\log L(\theta)) d\theta  &\leq o(n\epsilon_n^2)  
			\end{align*}
			By \eqref{e:Jn-0-order}, $\log J_n(\theta_0)=o(n\epsilon_n^2 \log n)$
			and by Lemma \ref{constant-order} and $-\log \epsilon_n=O(\log n)$
			$$\int q(\theta) \log J_n(\theta) d\theta \leq \epsilon_n \sqrt{n(1+\gamma^2)/2} \int \beta q(\beta)d\beta+o(n\epsilon_n^2)(\log 3\sqrt{2}-\log \epsilon_n)=o(n\epsilon_n^2)$$
			where we use $E_Q(\beta)=\exp(\log \beta_0+1/n) \to \beta_0$
		\end{proof}
		
		\section{}
		\label{append-B}
		\subsection{Proof of Theorem \ref{thm:post-convergence}}
		
		\noindent In this section, with dominating probability term is used to imply that under $\mathbb{P}_0^{(n)}$, the probability of the event goes to 1 as $n \to \infty$. 
		
		\begin{align}
			\nonumber	{\rm KL}(Q,\Pi(|X^{(n)}))&=\int q(\theta)\log q(\theta)d\theta-\int q(\theta) \log \pi(\theta|X^{(n)})d\theta\\
			\nonumber	&=\int q(\theta)\log q(\theta)d\theta- \int q(\theta) \log \frac{L(\theta)p(\theta)}{\int L(\theta)p(\theta)d\theta} d\theta\\
			\nonumber&={\rm KL}(Q,P)-\int \log  (L(\theta)/L(\theta_0))q(\theta) d\theta+\log \int (L(\theta)/L(\theta_0))  p(\theta) d\theta\\
			&= {\rm KL}(Q,P)+\int \log  (L(\theta_0)/L(\theta))q(\theta) d\theta+\log \int (L(\theta)/L(\theta_0))  p(\theta) d\theta
		\end{align}
		By Lemma \ref{q_p_relation}, ${\rm KL}(Q,P)=o(n\epsilon_n^2 \log n)\leq (C/3) n\epsilon_n^2 \log n$. By Lemma \ref{lem:E3}, with dominating probability $$\int  \log  (L(\theta_0)/L(\theta))q(\theta) d\theta \leq (C/3) n \epsilon_n^2 \log n$$
		for any $C>0$. By Lemma \ref{lem:E2},  with dominating probability $$\log \int (L(\theta)/L(\theta_0))  p(\theta) d\theta\leq (C/3) n \epsilon_n^2 \log n$$
		Therefore, with dominating probability, for any $C>0$,
		$$	{\rm KL}(Q,\Pi(|X^{(n)}))\leq Cn\epsilon_n^2$$
		Further,
		\begin{align*}
			\Pi(\mathcal{U}_{\varepsilon_n}^c|X^{(n)})&=\frac{\int_{\mathcal{U}_{\varepsilon_n}^c} L(\theta)p(\theta)d\theta}{\int L(\theta)p(\theta)d\theta}=\frac{\int_{\mathcal{U}_{\varepsilon_n}^c} (L(\theta)/L(\theta_0))p(\theta)d\theta}{\int (L(\theta)/L(\theta_0))p(\theta)d\theta}
		\end{align*}
		By Lemma \ref{lem:E1}, with dominating probability, for any $C>0$, as $n \to \infty $
		$$\int_{\mathcal{U}_{\varepsilon_n}^c} (L(\theta)/L(\theta_0))p(\theta)d\theta\leq \exp(-C_0n \varepsilon_n^2  )$$
		By Lemma \ref{lem:E2}, with dominating probability
		$$\int (L(\theta)/L(\theta_0))p(\theta)d\theta\geq \exp(-C n {\epsilon_n^2} \log n)$$
		Therefore, with dominating probability
		\begin{align*}
			\Pi(\mathcal{U}_{\varepsilon_n}^c|X^{(n)})&\leq \exp(-C_0 n\varepsilon_n^2(1-C/M_n)) \leq \exp(-C_1 n\varepsilon_n^2)
		\end{align*}
		for any $C_0>C_1/2$. This is because for $n$ sufficiently large $1-C/M_n>1/2$.
		
		\noindent This completes the proof.
		
		\subsection{Proof of Corollary \ref{thm: bayes-convergence}}
		
		\noindent By Lemma \ref{lem:E1}, with dominating probability, there exists $C_0(r)>0$ such that as $n \to \infty $,
		$$\int_{\mathcal{U}_{r\varepsilon_n}^c}
		(L(\theta)/L(\theta_0))p(\theta)d\theta\leq \exp(-C_0(r)  r^2 n\varepsilon_n^2  )$$
		Let us assume, \begin{align}
			\label{e:C0-val}
			C_0(r)\geq C_2/r \hspace{5mm}\text{for all $r>0$ for some constant $C_2>0$} 
		\end{align} 
		Numerical evidence for validity of this assumption been provided in Supplement Section A.3. However, the explicit theoretical derivation is technically involved and has been avoided in this paper. By Lemma \ref{lem:E2}, with dominating probability
		$$\int (L(\theta)/L(\theta_0))p(\theta)d\theta\geq \exp(-C n {\epsilon_n^2} \log n)$$
		
		\noindent Note, that
		\begin{align*}
			\Pi(\mathcal{U}_{r\varepsilon_n}^c|X^{(n)})&=\frac{\int_{\mathcal{U}_{r\varepsilon_n}^c} L(\theta)p(\theta)d\theta}{\int L(\theta)p(\theta)d\theta}=\frac{\int_{\mathcal{U}_{r\varepsilon_n}^c} (L(\theta)/L(\theta_0))p(\theta)d\theta}{\int (L(\theta)/L(\theta_0))p(\theta)d\theta}
		\end{align*}
		
		\noindent Therefore, with dominating probability 
		\begin{align*}
			\Pi(\mathcal{U}_{r\varepsilon_n}^c|X^{(n)})&\leq \exp(-C_2 r n\varepsilon_n^2)\exp(Cn\epsilon_n^2 \log n)
		\end{align*} 
		
		\noindent Following steps of proof of proposition 11 on page 2111 in \cite{van2011information},
		\begin{align*}
			\int e^{(C_2/2)n\varepsilon_n ||\theta-\theta_0||_2}    d\Pi(\mathcal{U}_{\varepsilon_n}^c|X^{(n)})&=\int_0^{\infty} e^{(C_2/2)n\varepsilon_n^2 r}\Pi(||\theta-\theta_0||_2\geq r\varepsilon_n|X^{(n)})dr 
		\end{align*}
		
		\noindent Therefore,
		\begin{align*}
			& \int e^{(C_2/2) n\varepsilon_n ||\theta-\theta_0||_2}    d\Pi(\mathcal{U}_{\varepsilon_n}^c|X^{(n)})\\
			&=\exp(Cn\epsilon_n^2 \log n)\int_{0}^\infty \exp((C_2/2)r n\varepsilon_n^2)\exp(-C_2 r n\varepsilon_n^2)dr\\
			&= \exp(Cn\epsilon_n^2 \log n) \int_0^{\infty} \exp(-(C_2/2) rn \varepsilon_n^2)dr= \frac{2}{C_2 n\varepsilon_n^2 }\exp(Cn\epsilon_n^2 \log n) \\
		\end{align*}
		This completes the proof.
	\end{appendix}
	
	\begin{acks}[Acknowledgments]
		The authors are grateful to Promit Ghosal and Sumit Mukherjee for their accomplishment in Ising model parameter estimation and kindly sharing codes of their work. The research is partially  supported by NSF-DMS 1945824 and NSF-DMS 1924724.
	\end{acks}

	\begin{supplement}
		\stitle{Supplement Materials}
		\sdescription{{\bf Appendix A.} Contains implementation details of the BBVI algorithm, plots of ELBO convergence for the BBVI and proof of Relation \eqref{e:C0-val}.}
	\end{supplement}
	
	\bibliographystyle{imsart-number} 
	\bibliography{paper-ref.bib}       

\end{document}